\documentclass[11pt,reqno]{amsart}

\pdfoutput = 1

\usepackage[english]{babel}
\usepackage{amsmath,amssymb,amsthm}
\usepackage{amsfonts}

\usepackage{array,dcolumn}
\usepackage{graphicx}
\usepackage[hyperfootnotes=true,
        pdffitwindow=true,
        plainpages=false,
        pdfpagelabels=true,
        pdfpagemode=UseOutlines,
        pdfpagelayout=SinglePage,
        pagebackref,
        hyperindex,]{hyperref}

\usepackage[T1]{fontenc}
\usepackage[utf8]{inputenc}
\usepackage[activate={true,nocompatibility},spacing,kerning]{microtype}
\usepackage[sc,osf]{mathpazo}
\microtypecontext{spacing=nonfrench}

 \calclayout

\newcommand{\mathsym}[1]{{}}
\newcommand{\unicode}[1]{{}}

\theoremstyle{plain}
\newtheorem{theorem}{Theorem}
\newtheorem{lemma}[theorem]{Lemma}

\newtheorem{proposition}[theorem]{Proposition}
\numberwithin{theorem}{section}

\theoremstyle{definition}

\theoremstyle{remark}
\newtheorem{remark}[theorem]{Remark}

\renewcommand{\theequation}{\thesection.\arabic{equation}}

\newcommand{\Z}{\mathbb Z}

\newcommand{\+}{\!+\!}

\newcommand{\half}{\tfrac{1}{2}}

\renewcommand{\leq}{\leqslant}
\renewcommand{\geq}{\geqslant}

\begin{document}


\title[Loop equations for the Laguerre and Jacobi $\beta$ ensembles]{Large $N$ expansions for the
Laguerre and Jacobi $\beta$ ensembles from the loop equations}
\author{Peter J. Forrester}
\address{School of Mathematics and Statistics, 
ARC Centre of Excellence for Mathematical
 and Statistical Frontiers,
University of Melbourne, Victoria 3010, Australia}
\email{pjforr@unimelb.edu.au}
\author{Anas A. Rahman}
\address{School of Mathematics and Statistics, 
ARC Centre of Excellence for Mathematical
 and Statistical Frontiers,
University of Melbourne, Victoria 3010, Australia}
\email{anas.rahman@live.com.au}
\author{Nicholas S. Witte}
\address{
Institute of Fundamental Sciences,
Massey University
Private Bag 11222
Palmerston North 4442
New Zealand}
\email{N.S.Witte@massey.ac.nz}

\date{\today}


\begin{abstract}
The $\beta$-ensembles of random matrix theory with classical weights have many special properties.
One is that the loop equations specifying the resolvent and corresponding multipoint correlators permit a
derivation at general order of the correlator via Aomoto's method from the theory of the Selberg
integral. We use Aomoto's method to derive the full hierarchy of loop equations for Laguerre and Jacobi
$\beta$ ensembles, and use these to systematically construct the explicit form of the $1/N$ expansion at
low orders. This allows us to give the explicit form of corrections to the global density, and allows various moments to be
computed, complementing results available in the literature motivated by problems in quantum transport.

\end{abstract}


\maketitle


\section{Introduction}\label{s1}
Many applications of random matrices, for example in theoretical physics \cite{GMW98,Be97} and multivariate statistics \cite{Mu82} make use of the explicit functional form for the eigenvalue probability density function (PDF). In this context, let $H$ be a Hermitian random matrix with eigenvalue/eigenvector decomposition $H=ULU^{\dagger}$, where $L$ is the diagonal matrix of eigenvalues $\{\lambda_j\}$ and $U$ the corresponding unitary matrix of normalised eigenvectors. Let the entries of $H$ be either real ($\beta=1$), complex ($\beta=2$) or real quaternion ($\beta=4$) with the latter represented as $2\times2$ matrices (see e.g. \cite[Eq. (1.20)]{Fo10}); the label $\beta$ -- sometimes referred to as the Dyson index -- is the number of independent real numbers in an off diagonal entry. A fundamental result of random matrix theory (see e.g. \cite[Prop. 1.3.4]{Fo10}) is that the volume form associated with $H$, $(\mathrm{d}H)=\prod_{i=1}^N\mathrm{d}x_{ii}\prod_{1\leq j<k\leq N}\prod_{s=1}^{\beta}\mathrm{d}x_{jk}^{(\beta)}$, which is the product of the differentials of the independent real numbers in each entry, transforms in terms of the eigenvalues and eigenvectors according to
\begin{equation}\label{1.1}
(\mathrm{d}H)=\prod_{1\leq j<k\leq N}|\lambda_k-\lambda_j|^{\beta}\,\prod_{j=1}^N\mathrm{d}\lambda_j\,\left(U^{\dagger}\mathrm{d}U\right).
\end{equation}

One sees immediately from \eqref{1.1} that the simplest eigenvalue PDFs coming from this setting are the functional forms
\begin{equation}\label{1.2}
\frac{1}{C_N}\prod_{l=1}^Nw(\lambda_l)\,\prod_{1\leq j<k\leq N}|\lambda_k-\lambda_j|^{\beta},
\end{equation}
where $C_N$ denotes the normalisation and $w(\lambda)$ is referred to as the weight function. Random matrices with eigenvalue PDF specified by \eqref{1.2} are said to belong to a $\beta$-ensemble, indexed by $\beta$ and $w$, and denoted by $ME_{\beta,N}[w]$ (here $ME$ stands for matrix ensemble). Within this class, the so-called classical weights
\begin{equation}\label{1.3}
w(x)=\left\{\begin{array}{ll}e^{-x^2},&\quad\textrm{Gaussian}\\x^{\alpha_1}e^{-x}\chi_{x>0},&\quad\textrm{Laguerre}\\x^{\alpha_1}(1-x)^{\alpha_2}\chi_{0<x<1},&\quad\textrm{Jacobi}\\1/\left((1+\mathrm{i}x)^{\alpha}(1-\mathrm{i}x)^{\bar{\alpha}}\right),&\quad\textrm{Cauchy}\end{array}\right.,
\end{equation}
where $\chi_A=1$ if $A$ is true and $\chi_A=0$ otherwise, exhibit some distinguished properties. One is that for all the cases \eqref{1.3}, $C_N$ in \eqref{1.2} can be evaluated as an explicit product of gamma functions for general $\beta>0$, by virtue of the Selberg integral \cite{Se44}, \cite{FW07p}, \cite[Ch. 4]{Fo10}. Another is that the averaged characteristic polynomial can, for general $\beta>0$, be evaluated in terms of the $N$\textsuperscript{th} degree polynomial associated with these weights, e.g. a Jacobi polynomial in the Jacobi case \cite{Ao87}.

For $\beta=1$, $2$ and $4$ the $\beta$-ensemble with Gaussian weight results from a PDF on the space of Hermitian matrices with real, complex and real quaternion entries proportional to $\exp(-\textrm{Tr}\,H^2)$. Analogous specifications of PDFs on Hermitian matrices implying the other three weights in \eqref{1.3} can be given. For general $\beta>0$ the $\beta$-ensembles with classical weights can be realised as the eigenvalue PDF of certain tridiagonal and Hessenberg matrices \cite{DE02}, \cite{KN04}; see also \cite{Fo10}.

Consider the Cauchy weight $\beta$-ensemble and set $\alpha=b_1+\mathrm{i}b_2+\beta(N-1)/2$. Under the change of variables $\lambda_l=\mathrm{i}(1-e^{\mathrm{i}\theta_l})/(1+e^{\mathrm{i}\theta_l})$, $-\pi<\theta_l\leq\pi$, corresponding to a stereographic projection from the real line to the unit circle, this transforms to
\begin{equation}\label{1.4}
\frac{1}{C_N}\prod_{l=1}^Ne^{b_2\theta_l}|1+e^{\mathrm{i}\theta_l}|^{2b_1}\,\prod_{1\leq j<k\leq N}|e^{\mathrm{i}\theta_k}-e^{\mathrm{i}\theta_j}|^{\beta}
\end{equation}
(see e.g. \cite[Eq. (3.124)]{Fo10}). The eigenvalue PDF (\ref{1.4}) specifies a $\beta$-ensemble for unitary matrices, and has been referred to as the circular Jacobi $\beta$-ensemble \cite[\S 3.9]{Fo10}.

In previous works by two of the present authors, a detailed analysis of the Gaussian $\beta$-ensemble \cite{WF14} and circular $\beta$-ensemble \eqref{1.4} \cite{WF15} was carried out in the so-called global regime. This is the large $N$ limit, taken simultaneous to a scaling of the eigenvalues, so that the eigenvalue support is a finite interval. The main method used in the analysis was the loop equation formalism \cite{Mi83}, \cite{AM90}, which provides for a systematic $1/N$ expansion of the spectral density. In the Gaussian case, but restricted to $\beta=1$, $2$ and $4$, this was supplemented by the derivation of a specific $3$\textsuperscript{rd} order (for $\beta=2$) and $5$\textsuperscript{th} order (for $\beta=1$ and $4$) linear differential equation satisfied by the spectral density. It is the purpose of the present paper to undertake an analogous study of the Laguerre and Jacobi $\beta$-ensembles, thereby making available in the literature the details of the loop equations, and their primary consequences, in all of the classical cases.

Particularly for the Dyson indices $\beta=1$, $2$ or $4$, moments of the spectral density for the Laguerre and Jacobi $\beta$-ensembles are relevant to analysing the conductance and Wigner delay time associated with quantum transport and quantum cavities \cite{VV08}, \cite{No08}, \cite{LV11}, \cite{MS11}, \cite{MS12}, \cite{CMSV16a}, \cite{CMSV16b}. However, we don't pursue this further here, but instead focus on demonstrating that there is a common mathematical framework underlying our ability to analyse the $\beta$-ensembles with classical weights using loop equations, namely integration by parts in the theory of the Selberg integral \cite{Ao87}, \cite[Ch. 4]{Fo10}.

After introducing the relevant correlators in \S 2, and notation for their $1/N$ expansion, \S 3 focuses first on the
derivation of the loop equations for the Laguerre $\beta$-ensemble using Aomoto's method from Selberg integral
theory, and then on detailing some of the consequences. The primary consequence is the determination of the
explicit functional form for the first few terms in the $1/N$ expansion of the resolvent of the spectral density, or equivalently the corresponding expansion of the moment generating function. Section 4 gives an extension to the Jacobi case. The
final step in the derivation of the loop equations in both the Laguerre and Jacobi cases requires reducing a hierarchy of equations involving $n$-point
correlators down to a form involving only connected correlators. This is done by induction, with the details
presented in Appendix A. In distinction to the Gaussian and Laguerre cases, the moments $m_k^{(J)}$ in the
Jacobi $\beta$-ensemble are rational functions rather than polynomials in $1/N$. Thus even for small $k$ their exact
form can only be accessed indirectly by a finite $1/N$ expansion. Nonetheless, methods based on Jack polynomials have 
been used in an earlier work to specify $m_1^{(J)}$ and $m_2^{(J)}$ \cite{MR15}. In Appendix B we similarly specify
$m_3^{(J)}$ and furthermore relate the $m_k^{(J)}$ with $\alpha_1, \alpha_2$ specialised to moments in the circular
$\beta$-ensemble known from \cite{WF15}.

\setcounter{equation}{0}
\section{Resolvent, Connected Correlators and Large $N$ Form}\label{s2}
Generally the eigenvalue density corresponding to \eqref{1.2}, $\rho_{(1)}(\lambda)$, is defined so that $\int_a^b\rho_{(1)}(\lambda)\mathrm{d}\lambda$ is equal to the expected number of eigenvalues between $a$ and $b$. Its Stieltjes transform -- also referred to as the resolvent -- is given by
\begin{equation}\label{2.1}
\overline{W}_1(x)=\int_{-\infty}^{\infty}\frac{\rho_{(1)}(\lambda)}{x-\lambda}\mathrm{d}\lambda.
\end{equation}
Introducing the notation $\langle\,\cdot\,\rangle_{ME_{\beta,N}[w]}$ to denote the average with respect to the $\beta$-ensemble weight $w$  \eqref{1.2}, this can be written
\begin{equation}\label{2.2}
\overline{W}_1(x)=\left\langle\sum_{j=1}^N\frac{1}{x-\lambda_j}\right\rangle_{ME_{\beta,N}[w]}.
\end{equation}
For a large class of weights $w$, there is a scale $x = c_N s$ such that in the variable $s$ and as $N \to \infty$ the eigenvalue support is
a finite interval, and moreover $\overline{W}_1(c_N s)$ can be expanded as a series in $1/N$ \cite{BG12},
\begin{equation}\label{2.2a}
c_N \overline{W}_1(c_N s) = N \sum_{l=0}^\infty {W_1^l(s) \over (N \sqrt{\kappa})^l}, \qquad \kappa = \beta/2.
\end{equation}

The average \eqref{2.2} is an example of a one-point correlator. To calculate $\{ W_n^l(s) \}$ via loop equations requires consideration of the corresponding general $n$-point correlator
\begin{equation}\label{2.3}
\overline{U}_n(x_1,\ldots,x_n)=\left\langle\sum_{j_1,\ldots,j_n=1}^N\frac{1}{(x_1-\lambda_{j_1})\cdots(x_n-\lambda_{j_n})}\right\rangle_{ME_{\beta,N}[w]}.
\end{equation}
Moreover, these quantities must be arranged to occur only in certain linear combinations $\overline{W}_n$, the connected components of $\overline{U}_n$, specified by
\begin{align}
\overline{W}_1(x)&=\overline{U}_1(x)\nonumber
\\ \overline{W}_2(x_1,x_2)&=\overline{U}_2(x_1,x_2)-\overline{U}_1(x_1)\overline{U}_1(x_2)\nonumber
\\ \overline{W}_3(x_1,x_2,x_3)&=\overline{U}_3(x_1,x_2,x_3)-\overline{U}_2(x_1,x_2)\overline{U}_1(x_3)-\overline{U}_2(x_1,x_2)\overline{U}_1(x_2)\nonumber
\\&\quad-\overline{U}_2(x_2,x_3)\overline{U}_1(x_1)+2\overline{U}_1(x_1)\overline{U}_1(x_2)\overline{U}_1(x_3)\label{2.4}
\end{align}
and in general
\begin{equation}\label{2.5}
\overline{W}_n(x_1,\ldots,x_n)=\sum_{m=1}^n\sum_G(-1)^{m-1}(m-1)!\prod_{j=1}^m\overline{U}_{|G_j|}(x_{g_j(1)},\ldots,x_{g_j(|G_j|)}),
\end{equation}
where the sum over $G$ is over all subdivisions of $\{1,\ldots,n\}$ into $m$ subsets $G_j=\{g_j(1),\ldots,g_j(|G_j|)\}$. It is also true that $\overline{U}_n$ can be written in terms of the $W_k$ according to the inductive formula (see e.g. \cite[pp. 8-9]{WF15})
\begin{align}
\overline{U}_n(x_1,J_n)=\overline{W}_n(x_1,J_n)+\sum_{\emptyset\neq J\subseteq J_n}\overline{W}_{n-|J|}(x_1,J_n\setminus J)\overline{U}_{|J|}(J),\label{2.6}
\\J_n=(x_2,\ldots,x_n),\quad J_1=\emptyset.\nonumber
\end{align}

The crucial feature of the connected correlators $\overline{W}_n$ is that, in distinction to the unconnected correlators $\overline{U}_n$, for large $N$ they permit an expansion decaying like $N^{2-n}$ \cite{BG12}
\begin{equation}\label{2.7}
c_N^n\overline{W}_n(c_Ns_1,\ldots,c_Ns_n)=N^{2-n}\kappa^{1-n}\sum_{l=0}^{\infty}\frac{W_n^l(s_1,\ldots,s_n)}{(N\sqrt{\kappa})^l},\quad\kappa=\frac{\beta}{2},
\end{equation}
in the global regime, the latter referring to the scaling $\lambda_j = c_j s$ 
so that the eigenvalue support is a finite interval as already introduced in
(\ref{2.2a}); in fact (\ref{2.2a}) corresponds to the case $n=1$.
The $W_n^l$ are independent of $N$, but may depend on $\kappa$. The structure of \eqref{2.7} leads to a closed system of equations for the determination of any one of the $W_n^l$, and in particular for $\{W_1^l\}$.

Of the quantities in (\ref{2.7}), $W_2^0(s_1,s_2)$ is special. Thus, as first derived in a loop equation study of Hermitian matrix models
\cite{AJM90}, for settings that the global eigenvalue density is supported on a single interval $(a,b)$, $W_2^0$ for general $\kappa$ is expected to exhibit the
universal form
\begin{equation}\label{2.7a}
W_2^0(s_1,s_2) = {s_1 s_2 - (a+b) (s_1 + s_2)/2 + ab \over
2 (s_1 - s_2)^2 \sqrt{(s_1 - a) ( s_1 - b) (s_2 - a) (s_2 - b)}} - {1 \over 2 (s_1 - s_2)^2}.
\end{equation}
Our results below indeed conform with (\ref{2.7a}) in all cases considered.

\setcounter{equation}{0}
\section{Loop Equations for the Laguerre $\beta$-Ensemble}\label{s3}
\subsection{Aomoto's Method}\label{s3.1}
Let $e_p(x_1,\ldots,x_n)=\sum_{1\leq j_1<\cdots<j_p\leq n}\prod_{l=1}^px_{j_l}$ denote the $p$\textsuperscript{th} elementary symmetric function. Aomoto \cite{Ao87} (see also \cite[\S4.6]{Fo10}) showed how integration by parts could be used to deduce a first order recurrence in $p$ for $\langle e_p\rangle_{ME_{\beta,N}[\lambda^{\alpha_1}(1-\lambda)^{\alpha_2}]}$. Here we will show how this same general strategy, applied to the averages
\begin{equation}\label{3.1}
\left\langle\sum_{j_1=1}^N\frac{\partial}{\partial\lambda_{j_1}}\sum_{j_2,\ldots,j_n=1}^N\frac{1}{(x_2-\lambda_{j_2})\cdots(x_n-\lambda_{j_n})}\right\rangle^+_{ME_{\beta,N}[\lambda^{\alpha_1}e^{-\lambda}]}
\end{equation}
and
\begin{equation}\label{3.2}
\left\langle\sum_{j_1=1}^N\frac{\partial}{\partial\lambda_{j_1}}\frac{1}{x_1-\lambda_{j_1}}\sum_{j_2,\ldots,j_n=1}^N\frac{1}{(x_2-\lambda_{j_2})\cdots(x_n-\lambda_{j_n})}\right\rangle^+_{ME_{\beta,N}[\lambda^{\alpha_1}e^{-\lambda}]}
\end{equation}
leads us to the general loop equation for the Laguerre $\beta$-ensemble. Here, the superscript $+$ indicates that the derivative operation $\frac{\partial}{\partial \lambda_{j_1}}$ is to act on all functions of $\lambda_{j_1}$ to the right, including those present in the PDF when viewing the average as an integral. In particular, one must take special care in the cases $j_p=j_1$ for $p\neq1$. We take this opportunity to replace the notation $\langle\,\cdot\,\rangle_{ME_{\beta,N}[\lambda^{\alpha_1}e^{-\lambda}]}$ with $\langle\,\cdot\,\rangle^{(L)}$.

\begin{proposition}\label{prop3.1}
Define $\overline{U}_0 = 1$. For $\alpha_1>0$ at least we have
\begin{align}
0&=-\chi_{n\neq1}\sum_{k=2}^n\frac{\partial}{\partial x_k}\overline{U}_{n-1}(x_2,\ldots,x_n)+\alpha_1\left\langle\sum_{j_1=1}^N\frac{1}{\lambda_{j_1}}\sum_{j_2,\ldots,j_n=1}^N\frac{1}{(x_2-\lambda_{j_2})\cdots(x_n-\lambda_{j_n})}\right\rangle^{(L)}\nonumber
\\&\quad-N\overline{U}_{n-1}(x_2,\ldots,x_n).\label{3.3}
\end{align}
\end{proposition}
\begin{proof}
By the fundamental theorem of calculus in the variable $\lambda_j$, since with the assumption $\alpha_1>0$ the integrand vanishes at $\lambda_j=0$, and it also vanishes as $\lambda_j\rightarrow\infty$, the average \eqref{3.1} equals zero. On the other hand we can perform the differentiation directly, using the explicit functional form following from \eqref{1.2} for $ME_{\beta,N}[x^\alpha_1e^{-x}]$, allowing us to conclude
\begin{align*}
0&=\sum_{k=2}^n\left\langle\sum_{j_1,\ldots,\hat{j_k},\ldots,j_n=1}^N\frac{\chi_{n\neq1}}{(x_2-\lambda_{j_2})\cdots(x_{k-1}-\lambda_{j_{k-1}})(x_k-\lambda_{j_1})^2(x_{k+1}-\lambda_{j_{k+1}})\cdots(x_n-\lambda_{j_n})}\right\rangle^{(L)}
\\&\quad+\beta\left\langle\sum_{j_1,\ldots,j_n=1}^N\frac{1}{(x_2-\lambda_{j_2})\cdots(x_n-\lambda_{j_n})}\sum_{\substack{p=1\\p\neq j_1}}^N\frac{1}{\lambda_{j_1}-\lambda_p}\right\rangle^{(L)}
\\&\quad+\alpha_1\left\langle\sum_{j_1,\ldots,j_n=1}^N\frac{1}{\lambda_{j_1}}\frac{1}{(x_2-\lambda_{j_2})\cdots(x_n-\lambda_{j_n})}\right\rangle^{(L)}-\left\langle\sum_{j_1,\ldots,j_n=1}^N\frac{1}{(x_2-\lambda_{j_2})\cdots(x_n-\lambda_{j_n})}\right\rangle^{(L)},
\end{align*}
where the notation $\hat{j_k}$ denotes that the series does not sum over $j_k$. Comparison with \eqref{3.3} and recalling \eqref{2.3} shows that the only remaining task is to show that the second average vanishes. This holds true since the quantity in the average is antisymmetric upon interchange of $\lambda_{j_1}$ and $\lambda_p$.
\end{proof}

Applying analogous reasoning, beginning with \eqref{3.2} gives a companion identity to \eqref{3.3}.
\begin{proposition}\label{prop3.2}
In the same setting as Proposition \ref{prop3.1}
\begin{align}
0&=-\frac{\partial}{\partial x_1}\overline{U}_n(x_1,\ldots,x_n)-\chi_{n\neq1}\sum_{k=2}^n\frac{\partial}{\partial x_k}\left\{\frac{\overline{U}_{n-1}(x_1,\ldots,\hat{x}_k,\ldots,x_n)-\overline{U}_{n-1}(x_2,\ldots,x_n)}{x_k-x_1}\right\}\nonumber
\\&\quad+\alpha_1\left\langle\sum_{j_1,\ldots,j_n=1}^N\frac{1}{\lambda_{j_1}(x_1-\lambda_{j_1})\cdots(x_n-\lambda_{j_n})}\right\rangle^{(L)}-\overline{U}_n(x_1,\ldots,x_n)\nonumber
\\&\quad+\beta\left\langle\sum_{j_1,\ldots,j_n=1}^N\frac{1}{(x_1-\lambda_{j_1})\cdots(x_n-\lambda_{j_n})}\sum_{\substack{p=1\\p\neq j_1}}^N\frac{1}{\lambda_{j_1}-\lambda_p}\right\rangle^{(L)},\label{3.4}
\end{align}
where the notation $\hat{x}_k$ in $\overline{U}_{n-1}(x_1,\ldots,\hat{x}_k,\ldots,x_n)$ denotes that $x_k$ is not an argument of $\overline{U}_{n-1}$.
\end{proposition}

We observe that upon use of the simple partial fraction formula
\begin{equation}\label{3.4a}
\frac{1}{\lambda_{j_1}(x_1-\lambda_{j_1})}=\frac{1}{x_1}\left(\frac{1}{\lambda_{j_1}}+\frac{1}{x_1-\lambda_{j_1}}\right)
\end{equation}
in the third line of \eqref{3.4}, the identity \eqref{3.3} can be used to express the relevant average entirely in terms of $\overline{U}$'s,
\begin{align}
&\alpha_1\left\langle\sum_{j_1,\ldots,j_n=1}^N\frac{1}{\lambda_{j_1}(x_1-\lambda_{j_1})\cdots(x_n-\lambda_{j_n})}\right\rangle^{(L)}\nonumber
\\&=\frac{\alpha_1}{x_1}\overline{U}_n(x_1,\ldots,x_n)+\chi_{n=1}\frac{N}{x_1}+\chi_{n\neq1}\frac{N}{x_1}\overline{U}_{n-1}(x_2,\ldots,x_n)\nonumber
\\&\quad+\chi_{n\neq1}\sum_{k=2}^n\frac{\partial}{\partial x_k}\overline{U}_{n-1}(x_2,\ldots,x_n)\label{3.5}
\end{align}

We're still faced with the task of identifying the final average in \eqref{3.4} in terms of the $\overline{U}$'s.
\begin{lemma}
With $\kappa=\beta/2$ we have
\begin{multline}\label{3.6}
\beta\left\langle\sum_{j_1,\ldots,j_n=1}^N\frac{1}{(x_1-\lambda_{j_1})\cdots(x_n-\lambda_{j_n})}\sum_{\substack{p=1\\p\neq j_1}}\frac{1}{\lambda_{j_1}-\lambda_p}\right\rangle^{(L)}
\\=\kappa\overline{U}_{n+1}(x_1,x_1,x_2,\ldots,x_n)+\kappa\frac{\partial}{\partial x_1}\overline{U}_n(x_1,\ldots,x_n).
\end{multline}
\end{lemma}
\begin{proof}
Interchanging $\lambda_{j_1}$ and $\lambda_p$ allows the replacement
\begin{equation*}
\frac{1}{(x-\lambda_{j_1})(\lambda_{j_1}-\lambda_p)}\mapsto\frac{1}{2}\frac{1}{\lambda_{j_1}-\lambda_p}\left(\frac{1}{x_1-\lambda_{j_1}}-\frac{1}{x_1-\lambda_{j_p}}\right)=\frac{1}{2}\frac{1}{(x_1-\lambda_{j_1})(x_1-\lambda_{j_p})}
\end{equation*}
in the LHS of \eqref{3.6}, and the RHS follows.
\end{proof}

Substituting \eqref{3.5} and \eqref{3.6} in \eqref{3.4} all terms now involve the correlators \eqref{2.3}, and we have
\begin{align}
0&=\left[(\kappa-1)\frac{\partial}{\partial x_1}+\left(\frac{\alpha_1}{x_1}-1\right)\right]\overline{U}_n(x_1,J_n)+ \frac{N}{x_1}\overline{U}_{n-1}(J_n)\nonumber
\\&\quad+\chi_{n\neq1}\sum_{k=2}^n\frac{\partial}{\partial x_k}\left\{\frac{\overline{U}_{n-1}(x_1,\ldots,\hat{x}_k,\ldots,x_n)-\overline{U}_{n-1}(J_n)}{x_1-x_k}+\frac{1}{x_1}\overline{U}_{n-1}(J_n)\right\}\nonumber
\\&\quad+\kappa\overline{U}_{n+1}(x_1,x_1,J_n),\label{3.7}
\end{align}
where $J_n$ is as in \eqref{2.6}. Moreover, this can be written in terms of the connected correlators $\{\overline{W}_k\}$ to give us the hierarchy of loop equations for the Laguerre $\beta$-ensemble.

\begin{proposition}\label{prop3.4}
With $\kappa=\beta/2$ and $J_n$ as in \eqref{2.6}, for $n\in\mathbb{Z}^+$ we have
\begin{align}
0&=\left[(\kappa-1)\frac{\partial}{\partial x_1}+\left(\frac{\alpha_1}{x_1}-1\right)\right]\overline{W}_n(x_1,J_n)+\chi_{n=1}\frac{N}{x_1}\nonumber
\\&\quad+\chi_{n\neq1}\sum_{k=2}^n\frac{\partial}{\partial x_k}\left\{\frac{\overline{W}_{n-1}(x_1,\ldots,\hat{x}_k,\ldots,x_n)-\overline{W}_{n-1}(J_n)}{x_1-x_k}+\frac{1}{x_1}\overline{W}_{n-1}(J_n)\right\}\nonumber
\\&\quad+\kappa\left[\overline{W}_{n+1}(x_1,x_1,J_n)+\sum_{J\subseteq J_n}\overline{W}_{|J|+1}(x_1,J)\overline{W}_{n-|J|}(x_1,J_n\setminus J)\right].\label{3.8}
\end{align}
\end{proposition}
\begin{proof}
In the case $n=1$, \eqref{3.7} reads
\begin{equation*}
0=\left[(\kappa-1)\frac{\partial}{\partial x_1}+\left(\frac{\alpha_1}{x_1}-1\right)\right]\overline{U}_1(x_1)+\frac{N}{x_1}+\kappa\overline{U}_2(x_1,x_1).
\end{equation*}
Recalling \eqref{2.4}, in terms of the $\overline{W}$'s this reads
\begin{equation}\label{3.9}
0=\left[(\kappa-1)\frac{\partial}{\partial x_1}+\left(\frac{\alpha_1}{x_1}-1\right)\right]\overline{W}_1(x_1)+\frac{N}{x_1}+\kappa\left[\overline{W}_2(x_1,x_1)+\overline{W}_1(x_1)\overline{W}_1(x_1)\right]
\end{equation}
which agrees with \eqref{3.8} for $n=1$.

For $n=2$, we again begin by replacing the $\overline{U}$'s by the $\overline{W}$'s according to \eqref{2.4}. We then subtract \eqref{3.9} multiplied by $\overline{W}_1(x_2)$ to arrive at \eqref{3.8}.

The general $n$ case of \eqref{3.8} can be deduced from \eqref{3.7} using induction on the cases $n-1,n-2,\ldots,1$, with $n=1$ and $n=2$ as the base cases. This requires use of \eqref{2.6}. The details are given in  Appendix \ref{appA}; see also \cite{Ra16}.
\end{proof}

\begin{remark}
For the Gaussian $\beta$-ensemble $ME_{\beta,N}[e^{-x^2}]$ the $n$\textsuperscript{th} loop equation is \cite{BEMN10}, \cite{BMS11}, \cite{WF14}
\begin{align}
0&=\left[(\kappa-1)\frac{\partial}{\partial x_1}-2x\right]\overline{W}_n(x_1,J_n)+\chi_{n=1}(2N)\nonumber
\\&\quad+\chi_{n\neq1}\sum_{k=2}^n\frac{\partial}{\partial x_k}\left\{\frac{\overline{W}_{n-1}(x_1,\ldots,\hat{x}_k,\ldots,x_n)-\overline{W}_{n-1}(J_n)}{x_1-x_k}\right\}\nonumber
\\&\quad+\kappa\left[\overline{W}_{n+1}(x_1,x_1,J_n)+\sum_{J\subseteq J_n}\overline{W}_{|J|+1}(x_1,J)\overline{W}_{n-|J|}(x_1,J_n\setminus J)\right].\label{3.10}
\end{align}
The structural similarity with \eqref{3.8} is evident.
\end{remark}

\subsection{Global Scaling and a Triangular System for $\{W_k^l\}$}\label{s3.2}
We see that the first loop equation \eqref{3.9} involves $\overline{W}_1$ and $\overline{W}_2$, while \eqref{3.8} shows that the $n$\textsuperscript{th} loop equation involves $\{\overline{W}_j\}_{j=1}^{n+1}$. Thus there are always more unknowns than equations. This circumstance changes if we consider expansions of the loop equations in powers of $1/N$. To be able to expand the $\overline{W}_j$ in this manner, we know from \eqref{2.7} that global scaling of the eigenvalues is required so that for large $N$ the eigenvalue density is supported on a compact interval. Since the Laguerre $\beta$-ensemble $ME_{\beta,N}[x^{\alpha_1}e^{-x}]$ has its leading order eigenvalue support on $(0,4N\kappa)$ \cite[Prop. 3.2.3]{Fo10}, we thus set $C_N=N\kappa$ in \eqref{2.7} and we further define the global scaled connected correlator $W_n$ by
\begin{equation}\label{3.11}
W_n(s_1,\ldots,s_n)=(N\kappa)^n\overline{W}_n(N\kappa s_1,\ldots,N\kappa s_n).
\end{equation}
In terms of this scaled correlator, the $n$\textsuperscript{th} loop equation reads
\begin{align}
0&=\left[\frac{1}{N}\left(1-\frac{1}{\kappa}\right)\frac{\partial}{\partial s_1}+\left(\frac{\alpha_1}{N\kappa s_1}-1\right)\right]W_n(s_1,\ldots,s_n)+\chi_{n=1}\frac{N}{s_1}\nonumber
\\&\quad+\frac{\chi_{n\neq1}}{N\kappa}\sum_{k=2}^n\frac{\partial}{\partial s_k}\left\{\frac{W_{n-1}(s_1,\ldots,\hat{s}_k,\ldots,s_n)-W_{n-1}(s_2,\ldots,s_n)}{s_1-s_k}+\frac{1}{s_1}W_{n-1}(s_2,\ldots,s_n)\right\}\nonumber
\\&\quad+\frac{1}{N}\left[W_{n+1}(s_1,s_1,s_2,\ldots,s_n)+\sum_{J\subseteq (s_2,\ldots,s_n)}W_{|J|+1}(s_1,J)W_{n-|J|}(s_1,(s_2,\ldots,s_n)\setminus J)\right].\label{3.12}
\end{align}
We now substitute the RHS of \eqref{2.7} for $W_n$ and equate like powers of $1/N$. This gives a hierarchy of equations which in a sense to be made precise subsequently, permits a unique solution.

\begin{proposition}\label{P13}
In the case $n=1$, we deduce from \eqref{3.12} that
\begin{equation}\label{3.13}
0=\left(W_1^0(s_1)\right)^2-W_1^0(s_1)+\frac{1}{s_1},
\end{equation}
and that for $l\geq1$
\begin{equation}\label{3.14}
0=W_2^{l-2}(s_1,s_1)+\sum_{k=0}^lW_1^k(s_1)W_1^{l-k}(s_1)+h\frac{\mathrm{d}}{\mathrm{d}s_1}W_1^{l-1}(s_1)+\frac{\alpha_1}{\sqrt{\kappa}s_1}W_1^{l-1}(s_1)-W_1^l(s_1),
\end{equation}
where $h:=\sqrt{\kappa}-1/\sqrt{\kappa}$ and $W_n^{-1}:=0$ for all $n\geq1$.

For $n\geq2$, we deduce from \eqref{3.12} that
\begin{align}
0&=-W_n^0(s_1,\ldots,s_n)+\sum_{J\subseteq(s_2,\ldots,s_n)}W_{|J|+1}^0(s_1,J)W_{n-|J|}^0(s_1,(s_2,\ldots,s_n)\setminus J)\nonumber
\\&\quad+\sum_{k=2}^n\frac{\partial}{\partial s_k}\left\{\frac{W_{n-1}^0(s_1,\ldots,\hat{s}_k,\ldots,s_n)-W_{n-1}^0(s_2,\ldots,s_n)}{s_1-s_k}+\frac{1}{s_1}W_{n-1}^0(s_2,\ldots,s_n)\right\}\label{3.15}
\end{align}
and that for $l\geq1$
\begin{align}
0&=h\frac{\partial}{\partial s_1}{W_n^{l-1}}(s_1,\ldots,s_n)+\frac{\alpha_1}{\sqrt{\kappa}s_1}W_n^{l-1}(s_1,\ldots,s_n)-W_n^l(s_1,\ldots,s_n)\nonumber
\\&\quad+\sum_{k=2}^n\frac{\partial}{\partial s_k}\left\{\frac{W_{n-1}^l(s_1,\ldots,\hat{s}_k,\ldots,s_n)-W_{n-1}^l(s_2,\ldots,s_n)}{s_1-s_k}+\frac{1}{s_1}W_{n-1}^l(s_2,\ldots,s_n)\right\}\nonumber
\\&\quad+W_{n+1}^{l-2}(s_1,s_1,s_2,\ldots,s_n)+\sum_{J\subseteq(s_2,\ldots,s_n)}\sum_{k=0}^lW_{|J|+1}^k(s_1,J)W_{n-|J|}^{l-k}(s_1,(s_2,\ldots,s_n)\setminus J).\label{3.16}
\end{align}
\end{proposition}

\begin{remark}\label{R14}
We see from \eqref{3.14} and \eqref{3.16} that knowledge of $W_{n'}^{l'}$ with the first two arguments repeated is required; for example \eqref{3.14} with $l=2$ needs $W_2^0(s_1,s_1)$ as input. On the other hand, inspection of \eqref{3.15} shows
\begin{equation}\label{3.17}
0=-W_2^0(s_1,s_2)+2W_2^0(s_1,s_2)W_1^0(s_1)+\frac{\partial}{\partial s_2}\left\{\frac{W_1^0(s_1)-W_1^0(s_2)}{s_1-s_2}+\frac{1}{s_1}W_1^0(s_2)\right\},
\end{equation}
so it is not possible to literally set $s_1=s_2$ due to the removable singularity. Instead, a limiting process must be adopted.
\end{remark}

\begin{remark}
In deducing the system of equations of Proposition \ref{P13} from \eqref{3.12} it has been assumed that $\alpha_1$ is of order unity. The case that $\alpha_1$ is of order $N$ is considered in \S\ref{s3.6}.
\end{remark}

\subsection{Solving for $\{W_1^k\}$}\label{s3.3}
The equation \eqref{3.13} is a quadratic and so uniquely determines $W_1^0$ up to the branch of the square root. Since from \eqref{2.1}, $\overline{W}_1(x)\sim N/x$ as $x\rightarrow\infty$, it follows from \eqref{2.7} that $W_1^0(s)\sim1/s$ as $s\rightarrow\infty$ thus uniquely determining the solution of \eqref{3.13} as
\begin{equation}\label{3.18}
W_1^0(s_1)=\half\left(1-\sqrt{1-\tfrac{4}{s_1}}\right).
\end{equation}
We can now use knowledge of $W_1^0$ to determine $W_1^1$ since \eqref{3.14} with $l=1$ reads
\begin{equation}\label{3.19}
0=2W_1^0(s_1)W_1^1(s_1)+h\frac{\mathrm{d}}{\mathrm{d}s_1}W_1^0(s_1)+\frac{\alpha_1}{\sqrt{\kappa}s_1}W_1^0(s_1)-W_1^1(s_1).
\end{equation}
Thus we have
\begin{equation}\label{3.20}
W_1^1(s_1)=\frac{\alpha_1}{2\sqrt{\kappa}}\left(\frac{1}{t_1}-\frac{1}{s_1}\right)-\frac{h}{t_1^2},
\end{equation}
where we have introduced the notation
\begin{equation}\label{3.21}
t_i:=s_i\sqrt{1-\tfrac{4}{s_i}}.
\end{equation}
Also, to solve \eqref{3.17} for $W_2^0(s_1,s_2)$ only requires knowledge of $W_1^0(s_1)$, allowing us to deduce
the validity of (\ref{2.7a}) with interval of support corresponding to $a=0$, $b=4$. Generally in (\ref{2.7a}) the singularity at
$s_1 = s_2$ is removable and we have
\begin{equation}\label{3.23a}
W_2^0(s_1,s_1) = {(b-a)^2 \over 16 (s_1 - b)^2 (s_1 - a)^2},
\end{equation}
which in the present setting gives
\begin{equation}\label{3.23}
W_2^0(s_1,s_1)=\frac{1}{t_1^4}.
\end{equation}

Setting $l=2$ in \eqref{3.14} we see that $W_1^2$ is determined by $W_2^0(s_1,s_1)$, $W_1^0(s_1)$ and $W_1^1(s_1)$, all of which are known. Thus substituting \eqref{3.23}, \eqref{3.20} and \eqref{3.18} shows
\begin{equation}\label{3.24}
W_1^2(s_1)=\frac{s_1}{t_1^5}+\frac{\alpha_1^2}{\kappa t_1^3}-\frac{h\alpha_1}{2\sqrt{\kappa}}\left(\frac{4-s_1}{t_1^3}+\frac{s_1^2}{t_1^4}\right)+h^2\frac{2s_1^2-3s_1}{t_1^5}.
\end{equation}

More generally, if $l$ is even we first use \eqref{3.15} to compute $W_{l/2+1}^0(s_1,\ldots,s_{l/2})$. Then we successively use the loop equation for $W_n^l$ \eqref{3.16} with $l\geq2$ to compute $W_n^l$ with $l$ increasing by $2$ and $n$ decreasing by $1$ at each step. If $l$ is odd, we use \eqref{3.16} with $l=1$ to compute $W_{(l+1)/2}^1(s_1,\ldots,s_{(l+1)/2})$ and repeat the above procedure to compute $W_n^l$. The following table specifies the order in which the first ten coefficients of $W_1$ are calculated in thirty six steps.
\begin{center}
\begin{tabular}{l l|c c c c c c c c c c c}
&$l$&0&1&2&3&4&5&6&7&8&9&10
\\$n$&$W_n^l$
\\ \hline 1&&1&2&4&6&9&12&16&20&25&30&36
\\2&&3&5&8&11&15&19&24&29&35&&
\\3&&7&10&14&18&23&28&34&&&&
\\4&&13&17&22&27&33&&&&&&
\\5&&21&26&32&&&&&&&&
\\6&&31&
\end{tabular}
\end{center}
Due to space restrictions, we record only $W_1^3(s_1)$,
\begin{align}
W_1^3(s_1)&=-h\left(\frac{s_1^4+6s_1^3-6s_1^2}{t_1^8}\right)-h^3\left(\frac{6s_1^4-12s_1^3+12s_1^2}{t_1^8}\right)+\frac{\alpha_1}{\sqrt{\kappa}}\left(\frac{s_1^3+s_1^2}{t_1^7}\right)\nonumber
\\&\quad-\frac{h^2\alpha_1}{\sqrt{\kappa}}\left(\frac{s_1^3-6s_1^2+8s_1}{t_1^6}-\frac{s_1^4+3s_1^3-3s_1^2}{t_1^7}\right)+\frac{h\alpha_1^2}{2\kappa}\left(\frac{1}{t_1^3}-\frac{s_1^3+6s_1^2-8s_1}{t_1^6}\right)\nonumber
\\&\quad+\frac{\alpha_1^3}{\kappa\sqrt{\kappa}}\left(\frac{s_1}{t_1^5}\right).\label{3.25}
\end{align}

Inspection of the explicit forms \eqref{3.18}, \eqref{3.20}, \eqref{3.24} and \eqref{3.25} of $W_1^0,\ldots,W_1^3$ reveals that each is a polynomial in $h=\sqrt{\kappa}-1/\sqrt{\kappa}$ and $\alpha_1/\sqrt{\kappa}$ of degree $l$. Moreover, for $h=\alpha_1/\sqrt{\kappa}=0$, $W_1^l(s_1)=0$ for $l$ odd. We observe too the large $s_1$ decay
\begin{equation}\label{3.26}
W_1^l(s_1)=O\left(\tfrac{1}{s_1^{l+1}}\right).
\end{equation}

\subsection{Smoothed Density}\label{s3.4}
Recalling \eqref{2.1} and \eqref{3.11} we have
\begin{equation}\label{3.26a}
W_1(s)=N\kappa\int_0^{\infty}\frac{\rho_{(1)}(N\kappa\lambda)}{s-\lambda}\mathrm{d}\lambda.
\end{equation}
If we knew $W_1(s)$ as an analytic function of $s$, $\rho_{(1)}$ would be reconstructed using the Sokhotski-Plemelj formula
\begin{equation}\label{3.27}
\rho_{(1)}(N\kappa s)=\frac{1}{2\pi\mathrm{i}}\lim_{\epsilon\rightarrow0^+}\left(W_1(s-\mathrm{i}\epsilon)-W_1(s+\mathrm{i}\epsilon)\right)
\end{equation}
for the inverse Stieltjes transform. However the loop equation formalism does not give an exact formula for $W_1(s)$, but rather terms in the $1/N$ expansion
\begin{equation}\label{3.28}
W_1(s)=N\sum_{l=0}^{\infty}\frac{W_1^l(s)}{(N\sqrt{\kappa})^l}.
\end{equation}
With the renormalised global density specified by $\tilde{\rho}_{(1)}(s;N)=\kappa\rho_{(1)}(N\kappa s)$, and thus the normalisation $\int_0^{\infty}\tilde{\rho}_{(1)}(s;N)\mathrm{d}s=1$, this suggests
\begin{equation}\label{3.29}
\tilde{\rho}_{(1)}(s;N)=\sum_{l=0}^{\infty}\frac{\tilde{\rho}_1^l(s)}{(N\sqrt{\kappa})^l},
\end{equation}
where the $\tilde{\rho}_1^l(s)$ are independent of $N$ and given by
\begin{equation}\label{3.30}
\tilde{\rho}_1^l(x)=\frac{1}{2\pi\mathrm{i}}\lim_{\epsilon\rightarrow0^+}\left(W_1^l(x-\mathrm{i}\epsilon)-W_1^l(x+\mathrm{i}\epsilon)\right).
\end{equation}

However, unlike \eqref{3.28}, \eqref{3.29} does not correspond to the true large $N$ asymptotic expansion of $\tilde{\rho}_{(1)}(s;N)$. Explicit calculation of the large $N$ asymptotic expansion of $\tilde{\rho}_{(1)}(s;N)$ for $\beta=1$, $2$ and $4$ \cite{FFG06e}, \cite{GFF05} and also for all even values of $\beta$ \cite{DF06}, shows that in addition to a series in $1/N$ there are also oscillatory terms with frequency proportional to $N$. Thus $\tilde{\rho}_{(1)}(s;N)$ as corresponding to \eqref{3.29} is to be interpreted as the smoothed global density, in which such oscillatory terms are absent.

To make use of \eqref{3.30}, we see from the explicit forms \eqref{3.18}, \eqref{3.20}, \eqref{3.24} and \eqref{3.25} of $W_1^0,\ldots,W_1^3$ that knowledge of the inverse Stieltjes transform of $(s-c)^{-m}$ is required for positive integer $m$, as is that of $(1-4/s)^{1/2-n}$ for integer $n$. One can readily verify that the former is equal to the distribution given by the $(m-1)$\textsuperscript{th} derivative of the Dirac delta function
\begin{equation}\label{Va1}
\frac{(-1)^{m-1}}{(m-1)!}\delta^{(m-1)}(x-c),
\end{equation}
while the latter is equal to
\begin{equation*}
\frac{(-1)^{n+1}}{\pi}\left(\frac{4}{x}-1\right)^{\half-n}\chi_{x\in(0,4)}.
\end{equation*}
Explicit functional forms for $\tilde{\rho}_1^0,\ldots,\tilde{\rho}_1^3$ can now be given, although to save space only the first three are noted.

\begin{proposition}\label{P19}
With $\tilde{\rho}_1^l(x)$ specified by \eqref{3.30} and corresponding to the coefficients in the $1/N$ expansion of the smoothed global density \eqref{3.29}, we have
\begin{align*}
\tilde{\rho}_1^0(x)=\frac{1}{2\pi}\sqrt{\tfrac{4}{x}-1}\,\chi_{x\in(0,4)},
\end{align*}
\begin{align*}
\tilde{\rho}_1^1(x)=\frac{h}{4}\left[\delta(x)-\delta(x-4)\right]+\frac{\alpha_1}{2\sqrt{\kappa}}\left[\frac{1}{\pi x}\left(\tfrac{4}{x}-1\right)^{-\tfrac{1}{2}}\,\chi_{x\in(0,4)}-\delta(x)\right],
\end{align*}
\begin{align*}
\tilde{\rho}_1^2(x)&=\frac{1}{\pi x^4}\left(\tfrac{4}{x}-1\right)^{-\tfrac{5}{2}}\,\chi_{x\in(0,4)}-\frac{{\alpha_1}^2}{\pi x^3\kappa}\left(\tfrac{4}{x}-1\right)^{-\tfrac{3}{2}}\,\chi_{x\in(0,4)}
\\&\quad+h^2\left[\frac{5}{\pi x^4}\left(\tfrac{4}{x}-1\right)^{-\tfrac{5}{2}}\,\chi_{x\in(0,4)}-\frac{2}{\pi x^3}\left(\tfrac{4}{x}-1\right)^{-\tfrac{3}{2}}\,\chi_{x\in(0,4)}\right]
\\&\quad+\frac{h\alpha_1}{2\sqrt{\kappa}}\left[\frac{1}{\pi x^2}\left(\tfrac{4}{x}-1\right)^{-\tfrac{1}{2}}\,\chi_{x\in(0,4)}+\delta^{'}(x-4)\right].
\end{align*}
\end{proposition}

\begin{remark}
For $l\geq2$ we see that the $\tilde{\rho}_1^l(x)$ contain non-integrable singularities at $x=0$ and $x=4$. The analogous quantities for the Gaussian $\beta$-ensemble also have non-integrable singularities at the end-points of their support, and in that setting, it was shown that all averages of monomials become well defined upon interpreting the integrals according to the value implied by the gamma function evaluation of the beta integral
\begin{equation}\label{3.31}
\int_0^1y^{a-1}(1-y)^{b-1}\mathrm{d}y=\frac{\Gamma(a)\Gamma(b)}{\Gamma(a+b)}.
\end{equation}
This prescription similarly works in the present setting; in particular one can use \eqref{3.31} in conjunction with the explicit functional forms of Proposition \ref{P19} to verify that for $p\in\mathbb{N}_0$
\begin{equation}\label{3.32}
\int_{-\infty}^{\infty}|x|^p\tilde{\rho}_{(1)}^l(x)\mathrm{d}x=0,\quad p<l,
\end{equation}
which is equivalent to \eqref{3.26}.
\end{remark}

\subsection{Moments}\label{s3.5}
Recalling the definition of $\tilde{\rho}_{(1)}(s;N)$ below \eqref{3.28}, \eqref{3.26a} reads
\begin{equation}\label{3.33}
W_1(s)=N\int_0^{\infty}\frac{\tilde{\rho}_{(1)}(s;N)}{s-\lambda}\mathrm{d}\lambda.
\end{equation}
Using the geometric series to expand $1/(1-\lambda/s)$ for large $|s|$ and integrating term by term shows
\begin{equation}\label{3.34}
\frac{1}{N}W_1(s)=\frac{1}{s}+\sum_{k=1}^{\infty}\frac{\tilde{m}_k}{s^{k+1}},
\end{equation}
where the $\tilde{m}_k$ are the moments of the scaled density,
\begin{equation}\label{3.35}
\tilde{m}_k=\int_0^{\infty}s^k\tilde{\rho}_{(1)}(s;N)\mathrm{d}s.
\end{equation}
A fundamental property of the $\tilde{m}_k= \tilde{m}_k^{(L)} $ in the Laguerre case derivable from Jack polynomial theory \cite{DE05}, \cite{MR15} is that they are polynomials of degree $k$ in $1/N$,
\begin{equation}\label{3.36}
\tilde{m}_k^{(L)} =\sum_{i=0}^ka_i^{(k)}N^{-i}.
\end{equation}
Substituting the expansion \eqref{3.28} in the LHS of \eqref{3.34}, and substituting \eqref{3.36} on the RHS, then equating like powers of $1/N$ shows
\begin{equation}\label{3.37}
W_1^l(s)=\kappa^{l/2}\sum_{k=l}^{\infty}\frac{a_l^{(k)}}{s^{k+1}}.
\end{equation}
Note in particular the consistency with \eqref{3.26}.

Since from \S\ref{s3.3} we know the explicit functional form of $W_1^l(s)$ for $l=0,\ldots,3$, we can expand for large $s$ and make use of \eqref{3.36} and \eqref{3.37} to read off the exact expressions for the small order moments.

\begin{proposition}\label{P21}
We have
\begin{align*}
\tilde{m}_0^{(L)}=1,
\end{align*}
\begin{align*}
\tilde{m}_1^{(L)}=1+\frac{1-\kappa+\alpha_1}{N\kappa},
\end{align*}
\begin{align*}
\tilde{m}_2^{(L)}=2+\frac{4-4\kappa+3\alpha_1}{N\kappa}+\frac{2-4\kappa+2\kappa^2+3\alpha_1-3\kappa\alpha_1+\alpha_1^2}{N^2\kappa^2},
\end{align*}
\begin{align*}
\tilde{m}_3^{(L)}&=5-\frac{2(8\kappa-5\alpha_1-8)}{N\kappa}+\frac{17-33\kappa+17\kappa^2+21\alpha_1-21\kappa\alpha_1+6\alpha_1^2}{N^2\kappa^2}\nonumber
\\&\quad+\frac{6-17\kappa+17\kappa^2-6\kappa^3+11\alpha_1-21\kappa\alpha_1+11\kappa^2\alpha_1+6\alpha_1^2-6\kappa\alpha_1^2+\alpha_1^3}{N^3\kappa^3}.
\end{align*}
\end{proposition}

\begin{remark}
Let $m_k^{(L)}=N^k\tilde{m}_k^{(L)}$, and write $m_k^{(L)}=m_k^{(L)}(\kappa,N,\alpha_1)$. Inspection of the above results reveals the symmetry
\begin{equation}\label{3.38}
m_k^{(L)}(\kappa,N,\alpha_1)=(-\kappa)^{-k}m_k^{(L)}(1/\kappa,-\kappa N,-\alpha_1/\kappa).
\end{equation}
This result can be deduced in general as a corollary of \cite[Th.2.11]{DE05}. Writing $W_1^k(s) = W_1^k(s;\kappa,\alpha_2)$,
it is in keeping with the symmetry
\begin{equation}\label{3.38a}
W_1^k(s;\kappa,\alpha_1) = (-1)^k W_1^k(s;1/\kappa,-\alpha_1/\kappa), 
\end{equation}
which we observe is satisfied by (\ref{3.20}), (\ref{3.24}) and (\ref{3.25})

Also, we make the empirical observation that for $\alpha_1=0$ the denominator of a term with numerator $(N\kappa)^l$ is a reciprocal polynomial in $\kappa$ with its zeros on the unit circle in the complex $\kappa$ plane. See \cite{WF14} for an analogous property in the case of the Gaussian $\beta$-ensemble. As a final remark, we note that the results of proposition \ref{P21} are also given in \cite[Appendix A]{MR15} with the identification $m_k^{(L)}=\kappa^{-k+1}M_k^{(L)}$.
\end{remark}

Inspection of Proposition \ref{P21} reveals the integer sequence $1,1,2,5,\ldots$ as the leading term in the $1/N$ expansion of $\tilde{m}_0^{(L)},\tilde{m}_1^{(L)},\tilde{m}_2^{(L)},\tilde{m}_3^{(L)},\ldots$, or in the notation of \eqref{3.36} as the value of $a_0^{(k)}$ ($k=0,1,2,3,\ldots$). On the other hand, \eqref{3.37} and \eqref{3.18} together tell us
\begin{equation*}
\half\left(1-\sqrt{1-\tfrac{4}{s}}\right)=\sum_{k=0}^{\infty}\frac{a_0^{(k)}}{s^{k+1}}.
\end{equation*}
The LHS is the generating function of the Catalan numbers, demonstrating the well known result (see e.g. \cite{PS11})
\begin{equation}\label{3.40a}
a_0^{(k)}=\frac{1}{k+1}\binom{2k}{k}.
\end{equation}

Proposition \ref{P21} shows that the term proportional to $1/N$ in the expansion of $\tilde{m}_k$, which in the notation of \eqref{3.36} is denoted $a_1^{(k)}$, is a linear function of $1/\kappa$ and $\alpha_1$. According to \eqref{3.37} and \eqref{3.20}
\begin{equation*}
\frac{\alpha_1}{2\kappa}\left(\frac{1}{s\sqrt{1-4/s}}-\frac{1}{s}\right)-\frac{(1-1/\kappa)}{s(s-4)}=\sum_{k=1}^{\infty}\frac{a_1^{(k)}}{s^{k+1}}
\end{equation*}
and thus
\begin{equation}\label{3.40b}
a_1^{(k)}=\frac{\alpha_1}{2\kappa}\left(\binom{2k}{k}-\delta_{k,0}\right)-\frac{(1-1/\kappa)}{4}(4^k-\delta_{k,0}).
\end{equation}
Explicit formulas for $a_2^{(k)}$ and $a_3^{(k)}$ can be deduced from \eqref{3.24} and \eqref{3.25}. Their structures are increasingly more complicated so we will not report on the actual expressions.

An alternative way to deduce \eqref{3.40a} and \eqref{3.40b} is to note from \eqref{3.29}, \eqref{3.35} and \eqref{3.37} that
\begin{equation*}
a_j^{(k)}=\int_0^4x^k\tilde{\rho}_1^j(x)\mathrm{d}x,
\end{equation*}
and to evaluate the integrals using the explicit form of the $\tilde{\rho}_1^j$ from Proposition \ref{P19}, together with \eqref{3.31}.

\subsection{The Case of $\alpha_1\propto N$}\label{s3.6}
Thus far we have investigated the case of $\alpha_1$ fixed independent of $N$.
We now turn our attention to the 
case of $\alpha_1\propto N$ by writing $\alpha_1=\tilde{\alpha}_1 N\kappa$ where $\tilde{\alpha}_1={\rm O}(1)$. 
This setting is relevant to the study of the so called proper delay times for scattering in a quantum cavity
\cite{BFB97}; see also \cite[\S 3.3.1]{Fo10}.
Up to the derivation of the loop equation \eqref{3.12}, the scaling of $\alpha_1$ is inconsequential. However Proposition \ref{P13}, where we equate like powers of $1/N$, is no longer valid,
as is to be replaced by the following hierarchy of equations.

\begin{proposition}\label{prop3.13}
In the case $n=1$, we deduce from \eqref{3.12} that
\begin{equation}\label{3.41}
0=\left(W_1^0(s_1)\right)^2+\left(\frac{\tilde{\alpha}_1}{s_1}-1\right)W_1^0(s_1)+\frac{1}{s_1},
\end{equation}
and that for $l\geq1$
\begin{equation}\label{3.42}
0=W_2^{l-2}(s_1,s_1)+\sum_{k=0}^lW_1^k(s_1)W_1^{l-k}(s_1)+h\frac{\mathrm{d}}{\mathrm{d}s_1}W_1^{l-1}(s_1)+\left(\frac{\tilde{\alpha}_1}{s_1}-1\right)W_1^l(s_1),
\end{equation}
where $h:=\sqrt{\kappa}-1/\sqrt{\kappa}$ and $W_n^{-1}:=0$ for all $n\geq1$.

For $n\geq2$, we deduce from \eqref{3.12} that
\begin{align}
0&=\left(\frac{\tilde{\alpha}_1}{s_1}-1\right)W_n^0(s_1,\ldots,s_n)+\sum_{J\subseteq(s_2,\ldots,s_n)}W_{|J|+1}^0(s_1,J)W_{n-|J|}^0(s_1,(s_2,\ldots,s_n)\setminus J)\nonumber
\\&\quad+\sum_{k=2}^n\frac{\partial}{\partial s_k}\left\{\frac{W_{n-1}^0(s_1,\ldots,\hat{s}_k,\ldots,s_n)-W_{n-1}^0(s_2,\ldots,s_n)}{s_1-s_k}+\frac{1}{s_1}W_{n-1}^0(s_2,\ldots,s_n)\right\}\label{3.43}
\end{align}
and that for $l\geq1$
\begin{align}
0&=h\frac{\partial}{\partial s_1}{W_n^{l-1}}(s_1,\ldots,s_n)+\left(\frac{\tilde{\alpha}_1}{s_1}-1\right)W_n^l(s_1,\ldots,s_n)\nonumber
\\&\quad+\sum_{k=2}^n\frac{\partial}{\partial s_k}\left\{\frac{W_{n-1}^l(s_1,\ldots,\hat{s}_k,\ldots,s_n)-W_{n-1}^l(s_2,\ldots,s_n)}{s_1-s_k}+\frac{1}{s_1}W_{n-1}^l(s_2,\ldots,s_n)\right\}\nonumber
\\&\quad+W_{n+1}^{l-2}(s_1,s_1,s_2,\ldots,s_n)+\sum_{J\subseteq(s_2,\ldots,s_n)}\sum_{k=0}^lW_{|J|+1}^k(s_1,J)W_{n-|J|}^{l-k}(s_1,(s_2,\ldots,s_n)\setminus J).\label{3.44}
\end{align}
\end{proposition}
These equations contain the same $W_{n'}^{l'}$ as the equations given in Proposition \ref{P13}, so they too form a triangular system. Moreover, this resulting system is solved by the technique outlined in \S\ref{3.3}. In particular, \eqref{3.41} tells us that
\begin{equation}\label{3.45}
W_1^0(s_1)=\frac{1}{2}\left(1-\frac{\tilde{\alpha}_1}{s_1}-\sqrt{\left(\frac{\tilde{\alpha}_1}{s_1}-1\right)^2-\frac{4}{s_1}}\right),
\end{equation}
while \eqref{3.42} with $l=1$ reveals that
\begin{equation}\label{3.46}
0=2W_1^0(s_1)W_1^1(s_1)+h\frac{\mathrm{d}}{\mathrm{d}s_1}W_1^0(s_1)+\left(\frac{\tilde{\alpha}_1}{s_1}-1\right)W_1^1(s_1).
\end{equation}
Thus we have
\begin{equation}\label{3.47}
W_1^1(s_1)=\frac{h}{2}\left[\frac{\tilde{\alpha}_1}{s_1u_1}+\left(\frac{\tilde{\alpha}_1^2}{s_1}-\tilde{\alpha}_1-2\right)\frac{1}{u_1^2}\right],
\end{equation}
where we define
\begin{equation}\label{3.48}
u_i:=s_i\sqrt{\left(\frac{\tilde{\alpha}_1}{s_i}-1\right)^2-\frac{4}{s_i}}.
\end{equation}

Use of (\ref{3.45}) in \eqref{3.43} with $n=2$ gives (\ref{2.7a}) with endpoints of the interval of support
\begin{equation}\label{ab}
a= (\sqrt{\tilde{\alpha}_1+1} - 1)^2, \quad
b =(\sqrt{\tilde{\alpha}_1+1} + 1 )^2,
\end{equation}
and these values substituted in (\ref{3.23a}) give
\begin{equation}\label{3.50}
W_2^0(s_1,s_1)=\left(\tilde{\alpha}_1+1\right)\frac{1}{u_1^4}.
\end{equation}

Setting $l=2$ in \eqref{3.42} and using the above then reveals that
\begin{align}
W_1^2(s_1)&=\left(\tilde{\alpha}_1+1\right)\frac{s_1}{u_1^5}+h^2\left[\left(\tilde{\alpha}_1+2\right)\frac{1}{u_1^3}+\left(\tilde{\alpha}_1+2-s_1\right)\frac{\tilde{\alpha}_1}{u_1^4}+\left(\tilde{\alpha}_1+1\right)\frac{5s_1}{u_1^5}\right].\label{3.51}
\end{align}
For the sake of comparison with \eqref{3.25}, we also present $W_1^3(s_1)$,
\begin{align}
W_1^3(s_1)&=h\left[\left(\tilde{\alpha}_1+1\right)\frac{5\tilde{\alpha}_1s_1}{u_1^7}+\frac{\tilde{\alpha}_1s_1}{2u_1^5}+\left(\tilde{\alpha}_1^2+14\left(\tilde{\alpha}_1+1\right)-\tilde{\alpha}_1s_1+32s_1\right)\frac{s_1}{2u_1^6}\right]\nonumber
\\&\quad+h\left(\tilde{\alpha}_1(s_1^2+3s_1)+-s_1^3+3s_1^2+2s_1\right)\frac{17s_1}{u_1^8}\nonumber
\\&\quad+h^3\left[\left(\tilde{\alpha}_1-3\right)\frac{2}{u_1^4}+\left(2\tilde{\alpha}_1+3s_1+4\right)\frac{\tilde{\alpha}_1}{u_1^5}+\left(\tilde{\alpha}_1+1\right)\frac{25\tilde{\alpha}_1s_1}{u_1^7}\right]\nonumber
\\&\quad+h^3\left(5\tilde{\alpha}_1^2(s_1+2)-\tilde{\alpha}_1s_1(5s_1-14)+30s_1^2-6s_1\right)\frac{1}{u_1^6}\nonumber
\\&\quad+h^3\left(\tilde{\alpha}_1(s_1+3)-s_1^2+3s_1+2\right)\frac{30s_1^2}{u_1^8}\label{3.52}
\end{align}

Like the $W_1^0,\ldots,W_1^3$ given in \eqref{3.18}, \eqref{3.20}, \eqref{3.24} and \eqref{3.25} of \S\ref{s3.3}, the $\alpha_1\propto N$ versions \eqref{3.45}, \eqref{3.47}, \eqref{3.51} and \eqref{3.52} are polynomials in $h=\sqrt{\kappa}-1/\sqrt{\kappa}$ and $\tilde{\alpha}_1$, with $W_1^l$ being a polynomial of order $l$ in $h$. Moreover, for $l$ odd, $W_1^l$ only contains terms that have odd powers of $h$ as coefficients, so that when $h=0$, $W_1^l=0$ for all odd $l$. This is very similar to what is seen in the Gaussian $\beta$-ensemble \cite{WF14,Ma14}.

%

We again use \eqref{3.36} and \eqref{3.37} to compute small order moments for the $\alpha_1\propto N$ case. Since the moments $\tilde{m}_k^{(L)}$ are polynomials in $1/N$ of degree $k$, the results can be anticipated from
Proposition \ref{P21}, and so this calculation serves more as a check on our working.
\begin{proposition}\label{prop3.14}
The results of Proposition \ref{P21} remain valid upon substitution of $\alpha_1=\tilde{\alpha}_1 N\kappa$.
\end{proposition}

In keeping with the discussion of the paragraph including \eqref{3.40a}, at each order of $1/N^p$ the coefficients are given implicitly by $W_1^p$ as the latter is their generating function. Thus from \eqref{3.45} and \eqref{3.37}, and with $s_1=\tilde{\alpha}_1t$, for terms independent of $N$
\begin{equation}\label{N1}
\frac{1}{2t}\left(t-1-t\sqrt{1-2\left(1+\tfrac{2}{\tilde{\alpha}_1}\right)\tfrac{1}{t}+\tfrac{1}{t^2}}\right)=\sum_{k=0}^{\infty}\frac{a_0^{(k)}}{(\tilde{\alpha}_1t)^{k+1}}.
\end{equation}
The LHS is the generating function for the Narayana polynomials \cite{St99}, reclaiming the well known result
(see e.g.~\cite[\S 3.4.2]{Fo10})
\begin{equation*}
a_0^{(k)}=\frac{1}{k}\sum_{j=1}^k\binom{k}{j}\binom{k}{j-1}(\tilde{\alpha}_1 + 1)^{j-1}.
\end{equation*}
Since here the support of the eigenvalue density is away from the origin, the negative moments are also well defined.
Consideration of the analogue of (\ref{N1}) shows that with $k \in \mathbb Z_{\ge 0}$
\begin{equation}\label{N2}
a_0^{(-k)} =  \tilde{\alpha}_1^{1 - 2k} a_0^{(k-1)}.
\end{equation}
This result is implied by \cite[Eq.~(39b)]{MS12}, and is shown therein to be relevant to the study of delay times in quantum dots.

According to \eqref{3.47} and \eqref{3.37}, for the terms of order $1/N$ in $\tilde{m}_k$ we have
\begin{multline*}
\frac{(1-1/\kappa)}{2}\Bigg\{\frac{1}{\tilde{\alpha}_1t^2\sqrt{1-2(1+2/\tilde{\alpha}_1)/t+1/t^2}}
\\+\left(\frac{\tilde{\alpha}_1}{t}-\tilde{\alpha}_1-2\right)\frac{1}{\tilde{\alpha}_1^2t^2(1-2(1+2/\tilde{\alpha}_1)/t+1/t^2)}\Bigg\}=\sum_{k=0}^{\infty}\frac{a_1^{(k)}}{(\tilde{\alpha}_1t)^{k+1}}.
\end{multline*}
Now, the Gegenbauer polynomials $C_n^{(\mu)}(x)$,
\begin{equation*}
C_n^{(\mu)}(x)=\sum_{k=0}^{\lfloor n/2\rfloor}(-1)^k\frac{\Gamma(n-k+\mu)}{\Gamma(\mu)k!(n-2k)!}(2x)^{n-2k}
\end{equation*}
obey the generating function formula
\begin{equation*}
\sum_{n=0}^{\infty}C_n^{(\mu)}(x)t^n=(1-2xt+t^2)^{-\mu}.
\end{equation*}
Using this in the previous formula we deduce
\begin{equation*}
a_1^{(k+1)}=\frac{\tilde{\alpha}_1^{k+1}(1-1/\kappa)}{2}\left\{C_k^{(1/2)}\left(1+\tfrac{2}{\tilde{\alpha}_1}\right)+C_{k-1}^{(1)}\left(1+\tfrac{2}{\tilde{\alpha}_1}\right)-\frac{\tilde{\alpha}_1+2}{\tilde{\alpha}_1}C_k^{(1)}\left(1+\tfrac{2}{\tilde{\alpha}_1}\right)\right\}
\end{equation*}
with $C_n^{(\mu)}(x)=0$ for $n\in\Z_{<0}$. A similar strategy gives for the negative moments
\begin{equation}\label{N3}
a_1^{(-k)} = {\tilde{\alpha}_1^{-k} (1 - 1/ \kappa) \over 2}
\bigg \{  -C_{k}^{(1/2)}\left(1+\tfrac{2}{\tilde{\alpha}_1}\right) +
C_{k}^{(1)}\left(1+\tfrac{2}{\tilde{\alpha}_1}\right) - \frac{\tilde{\alpha}_1+2}{\tilde{\alpha}_1}
C_{k-1}^{(1)}\left(1+\tfrac{2}{\tilde{\alpha}_1}\right) \bigg \}.
\end{equation}
This has been reported, albeit in a different form, in \cite[Eq.~(91) with $\tilde{\alpha}_1 = w - 1$]{MS12},
as has the corresponding generating function  \cite[Eq.~(96)]{MS12}. The latter can be read off from
the small $s_1$ expansion of  \eqref{3.47} with $u_1 \mapsto - u_1$ (this is necessary due to the change of
branch).\\

Knowledge of $W_1^0,\ldots,W_1^3$ from \eqref{3.45}, \eqref{3.47}, \eqref{3.51} and \eqref{3.52} allow us to compute the coefficients $\tilde{\rho}_1^0,\ldots,\tilde{\rho}_1^3$ of the large $N$ asymptotic expansion of the smoothed global density given by \eqref{3.29}. This is again accomplished through the Sokhotski-Plemelj formula \eqref{3.30}. Due to the structure of $W_1^0,\ldots,W_1^3$ in the $\alpha_1\propto N$ case, in addition to the inverse Stieltjes transform of $(s-c)^{-m}$ for positive integer $m$, we also need the inverse Stieltjes transform of $\left((\tilde{\alpha}_1-s)^2-4s\right)^{\half-n}$ for integer $n$. We recall that the former is given by (\ref{Va1})
and compute that the latter is
\begin{equation*}
\frac{(-1)^{n+1}}{\pi}\left(4x-(\tilde{\alpha}_1-x)^2\right)^{\half-n}\chi_{x\in(a,b)}
\end{equation*}
where $a,b$ are the endpoints of the support of the smoothed global density (\ref{ab}).
According to \eqref{3.29}, we can now give explicit functional forms for $\tilde{\rho}_1^0,\ldots,\tilde{\rho}_1^3$. To save space we only note the first three.

\begin{proposition}\label{prop3.15}
With $\tilde{\rho}_1^l(x)$ specified by \eqref{3.30} and corresponding to the coefficients in the $1/N$ expansion of the smoothed global density \eqref{3.29}, we have in the $\alpha_1\propto N$ case
\begin{align*}
\tilde{\rho}_1^0(x)= \frac{1}{2\pi x}\sqrt{(x - a)(b - x)}\chi_{x\in(a,b)},
\end{align*}

\begin{align*}
\tilde{\rho}_1^1(x)&=\frac{h}{2}\left[\frac{\tilde{\alpha}_1}{\pi x}\left((x-a)(b-x)\right)^{-\half}\chi_{x\in(a,b)}-\frac{\tilde{\alpha}_1+2}{4\sqrt{\tilde{\alpha}_1+1}}\left(\delta(x-b)-\delta(x-a)\right)\right]
\\&\quad+\frac{h}{2}\left[\frac{\tilde{\alpha}_1^2}{4\sqrt{\tilde{\alpha}_1+1}}\left(\frac{1}{b}\delta(x-b)-\frac{1}{a}\delta(x-a)\right)\right],
\end{align*}

\begin{align*}
\tilde{\rho}_1^2(x)&=\frac{(\tilde{\alpha}_1+1)x}{\pi}\left((a-x)(x-b)\right)^{-\tfrac{5}{2}}\chi_{x\in(a,b)}
\\&\,+\frac{h^2\chi_{x\in(a,b)}}{\pi}\left[5(\tilde{\alpha}_1+1)x\left((a-x)(x-b))^2\right)^{-\tfrac{5}{2}}-(\tilde{\alpha}_1+2)\left((a-x)(x-b)\right)^{-\tfrac{3}{2}}\right]
\\&\,-\frac{h^2\tilde{\alpha}_1(\tilde{\alpha}_1+2-x)}{16(\tilde{\alpha}_1+1)}\left[\delta'(x-b)+\delta'(x-a)+\frac{1}{2\sqrt{\tilde{\alpha}_1+1}}\left(\delta(x-b)-\delta(x-a)\right)\right].
\end{align*}
\end{proposition}

We remark that for $\tilde{\alpha}_1 = 0$, the above results reduce to the formulas of Proposition \ref{P19} in the
case $\alpha_1 = 0$. Also, at the endpoints of the support, $\tilde{\rho}_1^2$ contains non-integrable
singularities, as with the corresponding quantity in Proposition \ref{P19}.
In keeping with the discussion of \cite[End of \S 3.2]{WF14}, the integration by parts regularisation
$$
\int_a^b {x^p \over  ((a-x)(b-x))^{q/2}} \, dx =
{p \over (a + b) (1 - q/2)} 
\int_a^b {x^{p-1}  \over ((a-x)(b-x))^{q/2-1}} \, dx, \quad p,q \in \mathbb Z^+
$$
applied iteratively allows this and higher order terms in the expansion \eqref{3.29} to be integrated against monomials.

\subsection{The General Case of $\alpha_1={\rm O}(N)$}\label{s3.6}
In the interest of making contact with the problem of quantum transport as relevant to the Jacobi ensemble
considered in the next section, we need to consider $\alpha_1$ as a general linear function of $N$. We write $\alpha_1=\tilde{\alpha}_1N\kappa+\delta_1$ where $\tilde{\alpha}_1,\delta_1={\rm O}(1)$, and substitute this into \eqref{3.12} along with the RHS of \eqref{2.7} for $W_n$. Equating like powers of ${\rm O}(N)$ terms reveals that $W_1^0$ is again given be \eqref{3.45}. Moreover, equating like powers of ${\rm O}(1)$ with $n=1$ gives
\begin{equation}\label{3.60}
W_1^1(s_1)=\frac{h}{2}\left[\frac{\tilde{\alpha}_1}{s_1u_1}+\left(\frac{\tilde{\alpha}_1^2}{s_1}-\tilde{\alpha}_1-2\right)\frac{1}{u_1^2}\right]+\frac{\delta_1}{2\sqrt{\kappa}}\left[\frac{1}{u_1}-\frac{1}{s_1}-\frac{\tilde{\alpha}_1}{s_1u_1}\right].
\end{equation}
We remark that the first term is given by \eqref{3.47} and that for $\tilde{\alpha}_1 = 0$ and $\delta_1 \mapsto \alpha_1$
we recover (\ref{3.20}).

\pagebreak

\setcounter{equation}{0}
\section{Loop Equations for the Jacobi $\beta$-ensemble}\label{s4}
\subsection{Aomoto's Method}\label{s4.1}
To implement Aomoto's method in the case of the Jacobi $\beta$-ensemble, in addition to considering the averages \eqref{3.1} and \eqref{3.2} for the ensemble $(J):=ME_{\beta,N}[\lambda^{\alpha_1}(1-\lambda)^{\alpha_2}]$, we must also consider the third average
\begin{equation}\label{4.1}
\left\langle\sum_{j_1=1}^N\frac{\partial}{\partial\lambda_{j_1}}\lambda_{j_1}\sum_{j_2,\ldots,j_n=1}^N\frac{1}{(x_2-\lambda_{j_2})\cdots(x_n-\lambda_{j_n})}\right\rangle^{(J),+},
\end{equation}
where the superscript $+$ retains the meaning introduced with \eqref{3.1}.

\begin{proposition}\label{prop4.1}
For $\alpha_1,\alpha_2>0$ we have
\begin{align}
0&=-\chi_{n\neq1}\sum_{k=2}^n\frac{\partial}{\partial x_k}\overline{U}_{n-1}(x_2,\ldots,x_n)\nonumber
\\&\quad+\alpha_1\left\langle\sum_{j_1=1}^N\frac{1}{\lambda_{j_1}}\sum_{j_2,\ldots,j_n=1}^N\frac{1}{(x_2-\lambda_{j_2})\cdots(x_n-\lambda_{j_n})}\right\rangle^{(J)}\nonumber
\\&\quad-\alpha_2\left\langle\sum_{j_1=1}^N\frac{1}{1-\lambda_{j_1}}\sum_{j_2,\ldots,j_n=1}^N\frac{1}{(x_2-\lambda_{j_2})\cdots(x_n-\lambda_{j_n})}\right\rangle^{(J)}.\label{4.2}
\end{align}
\end{proposition}
\begin{proof}
The derivation, starting with the Jacobi version of \eqref{3.1}, is analogous to that for Proposition \ref{prop3.1}.
\end{proof}

\begin{proposition}\label{P23}
For $\alpha_1,\alpha_2>0$, the identity \eqref{3.4} with all averages now with respect to the Jacobi $\beta$-ensemble again holds true, except that the term $-\overline{U}_n(x_1,\ldots,x_n)$ is to be replaced by
\begin{equation}\label{4.2a}
-\alpha_2\left\langle\sum_{j_1,\ldots,j_n=1}^N\frac{1}{(1-\lambda_{j_1})(x_1-\lambda_{j_1})\cdots(x_n-\lambda_{j_n})}\right\rangle^{(J)}.
\end{equation}
\end{proposition}
\begin{proof}
This follows from the Jacobi version of \eqref{3.2}, analogous to \eqref{3.4}.
\end{proof}

Applying integration by parts to \eqref{4.1} gives a third identity, which determines \eqref{4.2} in terms of the $\overline{U}$'s.
\begin{proposition}
For $\alpha_1,\alpha_2>0$ we have
\begin{align}
0&=-\chi_{n\neq1}\left((n-1)\overline{U}_{n-1}(x_2,\ldots,x_n)+\sum_{k=2}^nx_k\frac{\partial}{\partial x_k}\overline{U}_{n-1}(x_2,\ldots,x_n)\right)\nonumber
\\&\quad+\left((\alpha_1+\alpha_2+1)N+\frac{\beta N(N-1)}{2}\right)\overline{U}_{n-1}(x_2,\ldots,x_n)\nonumber
\\&\quad-\alpha_2\left\langle\sum_{j_1,\ldots,j_n=1}^N\frac{1}{1-\lambda_{j_1}}\frac{1}{(x_2-\lambda_{j_2})\cdots(x_n-\lambda_{j_n})}\right\rangle^{(J)},\label{4.3}
\end{align}
where $\overline{U}_0:=1$.
\end{proposition}

Using \eqref{4.2} and \eqref{4.3} together allows all quantities in the identity of Proposition \ref{P23}, after first making use of \eqref{3.6} and the partial fraction formulas \eqref{3.4a} and
\begin{equation*}
\frac{1}{(1-\lambda_{j_1})(x_1-\lambda_{j_1})}=\frac{1}{1-x_1}\left(\frac{1}{x_1-\lambda_{j_1}}-\frac{1}{1-\lambda_{j_1}}\right),
\end{equation*}
to be written in terms of the $\overline{U}$'s. This then gives the Jacobi analogue of \eqref{3.7},
\begin{align}
0&=\left((\kappa-1)\frac{\partial}{\partial x_1}+\left(\frac{\alpha_1}{x_1}-\frac{\alpha_2}{1-x_1}\right)\right)\overline{U}_n(x_1,J_n)\nonumber
\\&\quad+\chi_{n\neq1}\sum_{k=2}^n\frac{\partial}{\partial x_k}\left\{\frac{\overline{U}_{n-1}(x_1,\ldots,\hat{x}_k,\ldots,x_n)-\overline{U}_{n-1}(J_n)}{x_1-x_k}+\frac{1}{x_1}\overline{U}_{n-1}(J_n)\right\}\nonumber
\\&\quad+\frac{1}{x_1(1-x_1)}\left[(\alpha_1+\alpha_2+1)N+\kappa N(N-1)\right]\overline{U}_{n-1}(J_n)\nonumber
\\&\quad-\frac{1}{x_1(1-x_1)}\chi_{n\neq1}\left((n-1)\overline{U}_{n-1}(J_n)+\sum_{k=2}^nx_k\frac{\partial}{\partial x_k}\overline{U}_{n-1}(J_n)\right)\nonumber
\\&\quad+\kappa\overline{U}_{n+1}(x_1,x_1,J_n)\label{4.4}
\end{align}
where $J_n$ is as in \eqref{2.6} and $\overline{U}_0:=1$.

Following the proof of Proposition \ref{prop3.4} this can be rewritten in terms of the connected correlators $\{\overline{W}_k\}$, so giving the hierarchy of loop equations for the Jacobi $\beta$-ensemble. The approach outlined in Appendix \ref{appA}, \cite{Ra16} applies.
\begin{proposition}
With $\kappa=\beta/2$ and $J_n$ as in \eqref{2.6}, for $n\in\mathbb{Z}^+$ we have
\begin{align}
0&=\left((\kappa-1)\frac{\partial}{\partial x_1}+\left(\frac{\alpha_1}{x_1}-\frac{\alpha_2}{1-x_1}\right)\right)\overline{W}_n(x_1,J_n)-\frac{n-1}{x_1(1-x_1)}\overline{W}_{n-1}(J_n)\nonumber
\\&\quad+\frac{\chi_{n=1}}{x_1(1-x_1)}\left[(\alpha_1+\alpha_2+1)N+\kappa N(N-1)\right]-\frac{\chi_{n\neq1}}{x_1(1-x_1)}\sum_{k=2}^nx_k\frac{\partial}{\partial x_k}\overline{W}_{n-1}(J_n)\nonumber
\\&\quad+\chi_{n\neq1}\sum_{k=2}^n\frac{\partial}{\partial x_k}\left\{\frac{\overline{W}_{n-1}(x_1,\ldots,\hat{x}_k,\ldots,x_n)-\overline{W}_{n-1}(J_n)}{x_1-x_k}+\frac{1}{x_1}\overline{W}_{n-1}(J_n)\right\}\nonumber
\\&\quad+\kappa\left[\overline{W}_{n+1}(x_1,x_1,J_n)+\sum_{J\subseteq J_n}\overline{W}_{|J|+1}(x_1,J)\overline{W}_{n-|J|}(x_1,J_n\setminus J)\right].\label{4.5}
\end{align}
\end{proposition}

We remark that both the loop equations (\ref{3.10}) for the Gaussian $\beta$-ensemble, and
(\ref{3.9}) for the Laguerre $\beta$-ensemble, can be obtained from (\ref{4.5}) via a limiting process. The latter is simply the same limiting process as is required to obtain the Gaussian and Laguerre weights from the Jacobi weight
as specified by (\ref{1.3}). Thus for the Gaussian weight, map $x_l \mapsto (1 - x_l/\alpha)/2$,
$\overline{W}_n \mapsto (-2 \alpha)^n \overline{W}_n$, set $\alpha_1 = \alpha_2 = \alpha^2$ and equate
coefficients of $(-2 \alpha)^{n+1}$ for $\alpha \to \infty$ to deduce (\ref{3.10}) from (\ref{4.5}). For the
Laguerre weight map $x_l \mapsto x_l/\alpha$, $\overline{W}_n \mapsto \alpha^n \overline{W}_n$, set
$\alpha_1 = \alpha$ and equate coefficients of $\alpha^{n+1}$ for $\alpha \to \infty$ to deduce (\ref{3.8}) from (\ref{4.5}).

\subsection{Triangular Systems for $\{W_k^l\}$}\label{s4.2}
Unlike the Laguerre case, the Jacobi weight is supported on a finite interval. Thus we do not need to introduce a scaling for the density or connected correlators, and instead expand the connected correlators as
\begin{equation*}
\overline{W}_n(s_1,\ldots,s_n)=N^{2-n}\kappa^{1-n}\sum_{l=0}^{\infty}\frac{W_n^l(s_1,\ldots,s_n)}{(N\kappa)^l}.
\end{equation*}
We substitute this into \eqref{4.5} and equate like powers of $1/N$ while considering the following three cases regarding $\alpha_1$ and $\alpha_2$:
\begin{enumerate}
\item $\alpha_1,\alpha_2={\rm O}(1)$
\item $\alpha_1=\tilde{\alpha}_1N\kappa+\delta_1$ with $\tilde{\alpha}_1,\delta_1,\alpha_2={\rm O}(1)$
\item $\alpha_1=\tilde{\alpha}_1N\kappa+\delta_1,\,\alpha_2=\tilde{\alpha}_2N\kappa+\delta_2$ where $\tilde{\alpha}_1,\tilde{\alpha}_2,\delta_1,\delta_2={\rm O}(1)$.
\end{enumerate}
We do not consider the case where $\alpha_1={\rm O}(1)$ while $\alpha_2={\rm O}(N)$ since $\alpha_1$ and $\alpha_2$ are dual under the change of variables $s_i\mapsto 1-s_i$. Also, for lack of present applicability, we do not report the cases that both $n\geq2$ and $l\geq1$.

\begin{proposition}\label{prop4.5}
Throughout this proposition, if $\alpha_i={\rm O}(1)$ we set $\tilde{\alpha}_i=0$, else we set $\alpha_i=\delta_i$.

In the case $n=1$, we deduce from \eqref{4.5} that
\begin{equation}\label{4.7}
0=\left(W_1^0(s_1)\right)^2+\left(\frac{\tilde{\alpha}_1}{s_1}-\frac{\tilde{\alpha}_2}{1-s_1}\right)W_1^0(s_1)+\frac{\tilde{\alpha}_1+\tilde{\alpha}_2+1}{s_1(1-s_1)},
\end{equation}
and that
\begin{align}
0&=\left((\kappa-1)\frac{\partial}{\partial s_1}+\left(\frac{\alpha_1}{s_1}-\frac{\alpha_2}{1-s_1}\right)\right)W_1^0(s_1)+\left(\frac{\tilde{\alpha}_1}{s_1}-\frac{\tilde{\alpha}_2}{1-s_1}\right)W_1^1(s_1)\nonumber
\\&\quad+\frac{\alpha_1+\alpha_2+1-\kappa}{s_1(1-s_1)}+2W_1^0(s_1)W_1^1(s_1),\label{4.8}
\end{align}
while for $l\geq2$
\begin{align}
0&=\left((\kappa-1)\frac{\partial}{\partial s_1}+\left(\frac{\alpha_1}{s_1}-\frac{\alpha_2}{1-s_1}\right)\right)W_1^{l-1}(s_1)+\left(\frac{\tilde{\alpha}_1}{s_1}-\frac{\tilde{\alpha}_2}{1-s_1}\right)W_1^l(s_1)\nonumber
\\&\quad+\kappa W_2^{l-2}(s_1,s_1)+\sum_{k=0}^lW_1^k(s_1)W_1^{l-k}(s_1).\label{4.9}
\end{align}

For $n\geq2$, we deduce from \eqref{4.5} that
\begin{align}
0&=\frac{n-1}{s_1(1-s_1)}W_{n-1}^0(s_2,\ldots,s_n)+\frac{1}{s_1(1-s_1)}\sum_{k=2}^ns_k\frac{\partial}{\partial s_k}W_{n-1}^0(s_2,\ldots,s_n)\nonumber
\\&\quad-\sum_{k=2}^n\frac{\partial}{\partial s_k}\left\{\frac{W_{n-1}^0(s_1,\ldots,\hat{s}_k,\ldots,s_n)-W_{n-1}^0(s_2,\ldots,s_n)}{s_1-s_k}+\frac{1}{s_1}W_{n-1}^0(s_2,\ldots,s_n)\right\}\nonumber
\\&\quad-\left(\frac{\tilde{\alpha}_1}{s_1}-\frac{\tilde{\alpha}_2}{1-s_1}\right)W_n^0(s_1,\ldots,s_n)-\sum_{J\subseteq (s_2,\ldots,s_n)}W_{|J|+1}^0(s_1,J)W_{n-|J|}^0(s_1,(s_2,\ldots,s_n)\setminus J).\label{4.10}
\end{align}
\end{proposition}
Using these equations, we can now compute the resolvent coefficients $W_n^l$, which further allow us to compute the large $N$ expansion of the density and the moments. Indeed, these equations form a triangular system, so we can compute these quantities up to whatever order we desire. We treat each of the three cases outlined in Proposition \ref{prop4.5} separately, specifying $W_1^0$, $W_1^1$ and $W_2^0$.

\subsection{Parameters $\alpha_1, \alpha_2$ of Order Unity}\label{s4.3}

\begin{proposition}\label{PSa} Suppose $\alpha_1, \alpha_2$ are fixed independent of $N$. We have
\begin{align}\label{4.12}
W_1^0(s_1) &=\frac{1}{\sqrt{s_1(s_1-1)}}, \\
W_1^1(s_1) & =\frac{\kappa-1-2\alpha_1}{4s_1}+\frac{\kappa-1-2\alpha_2}{4(s_1-1)}+\frac{\alpha_1+\alpha_2+1-\kappa}{2\sqrt{s_1(s_1-1)}} \label{V1}
\end{align}
and $W_2^0(s_1,s_2)$ is given by (\ref{2.7a}) with $a=0$, $b=1$.
\end{proposition}

\begin{proof}
In the case that $\alpha_1, \alpha_2$ are fixed, \eqref{4.7} reads 
\begin{equation*}
0=\left(W_1^0(s_1)\right)^2+\frac{1}{s_1(1-s_1)},
\end{equation*}
which implies the first equation in (\ref{4.12}) upon using the requirement $W_1^0(s_1)\sim 1/s_1$ as $s_1\rightarrow\infty$ to fix the branch.
This result substituted into \eqref{4.8} gives the second equation in (\ref{4.12}), while substituting it into 
\eqref{4.10} with $n=2$ gives $W_2^0(s_1,s_2)$.
\end{proof}

Using this data and the approach outlined in \S\ref{s3}, we compute
\begin{equation}\label{4.15}
m_1^{(J)} =\frac{N}{2}+\frac{\alpha_1-\alpha_2}{4\kappa}+{\rm O}\Big ( {1 \over N}\Big ).
\end{equation}
This can be compared against the exact expression for finite $N$ \cite[Eq.~(B.7a)]{MR15}
\begin{equation}\label{4.15a}
m_1^{(J)} = {N((N-1) \kappa + \alpha_1 + 1) \over 2 (N-1) \kappa + \alpha_1 + \alpha_2 + 2},
\end{equation}
which indeed exhibits the large $N$ expansion (\ref{4.15}).
We remark that in distinction to the expansion \eqref{3.36}, the $k$\textsuperscript{th} moment in general is no longer a polynomial of degree $k$ in $1/N$, but rather an infinite series \cite[Th.~5.1]{DP12}, as exhibited by the large $N$ expansion of (\ref{4.15a}). In fact each $m_k^{(J)}$ is a rational function in $N$ of degree $2k$ in the
numerator, and $2k-1$ in the denominator, the case $k=1$ being  given explicitly by
(\ref{4.15a}) and $k=2$ by \cite[Eq.~(B.7b)]{MR15}; these exact expressions are supplemented by us presenting
the explicit functional form of the case $k=3$ in Appendix B below.  Also,
with $m_k^{(J)} =  m_k^{(J)}(\kappa, N, \alpha_1, \alpha_2)$ the results of \cite{DP12} imply the
symmetry\footnote{V.~Gorin has informed us that this result can also be deduced from results contained in
the Appendix A of \cite{FLD16} written by A.~Borodin and V.~Gorin.}
\begin{equation}\label{4.15b}
m_k^{(J)} =  m_k^{(J)}(\kappa, N, \alpha_1, \alpha_2) = (-\kappa)^k m_k^{(J)}(1/\kappa, - \kappa N, -\alpha_1/\kappa, -\alpha_2/\kappa)
\end{equation}
(cf.~(\ref{3.28})), which in turn requires that $W_1^k(s) = W_1^k(s;\kappa,  \alpha_1, \alpha_2)$ exhibit the symmetry
\begin{equation}\label{4.15c}
 W_1^k(s;\kappa, \alpha_1, \alpha_2) = (-1)^k W_1^k(s;1/\kappa, - \alpha_1/\kappa, -\alpha_2/\kappa).
\end{equation} 
This last equation is indeed  a feature of (\ref{V1}).

We expand the density $\rho_{(1)}$ corresponding to the resolvent $\overline{W}_1$ as
\begin{equation}\label{4.16}
\rho_{(1)}(s)=N\sum_{l=0}^{\infty}\frac{\rho_1^l(s)}{(N\kappa)^l},
\end{equation}
where the $\rho_1^l$ are independent of $N$ and are given by the Sokhotski-Plemelj formula
\begin{equation}\label{4.17}
\rho_1^l(s)=\frac{1}{2\pi\mathrm{i}}\lim_{\epsilon\rightarrow0^+}\left(W_1^l(x-\mathrm{i}\epsilon)-W_1^l(x+\mathrm{i}\epsilon)\right)
\end{equation}
(cf.~(\ref{3.30})).
We recall that the inverse Stieltjes transform of $(s-c)^{-m}$ for positive integer $m$ is given by (\ref{Va1})
and further note that the inverse Stieltjes transform of $\left(s(s-1)\right)^{1/2-n}$ for integer $n$ is given by
\begin{equation*}
\frac{(-1)^{n+1}}{\pi}\left(x(1-x)\right)^{\half -n}\chi_{x\in(0,1)}.
\end{equation*}
Thus,
\begin{align}
\rho_1^0(x)&=\frac{1}{\pi\sqrt{x(1-x)}} \chi_{x \in (0,1)},\label{4.18}
\\ \rho_1^1(x)&=\frac{\kappa-1-2\alpha_1}{4}\delta(x)+\frac{\kappa-1-2\alpha_2}{4}\delta(x-1)+\frac{\alpha_1+\alpha_2+1-\kappa}{2\pi\sqrt{x(1-x)}} \chi_{x \in (0,1)}.\label{4.19}
\end{align}

If we expand the moments
\begin{equation}\label{Mn}
m_k^{(J)} = N \sum_{i=0}^\infty b_i^{(k)} (N \kappa)^{-i}
\end{equation}
the fact that
$
{1 \over \pi} \int_0^1 {x^k \over \sqrt{x (1 - x)}} \, dx = 2^{-2k} \binom{2k}{k}
$
used in (\ref{4.18}) and (\ref{4.19})  tells us that
\begin{equation}
b_0^{(k)} = 2^{-2k} \binom{2k}{k}, \: (k \ge 0) \quad
b_1^{(k)} =  {\kappa - 1 - 2 \alpha_2 \over 4} + {\alpha_1 + \alpha_2 + 1 - \kappa \over 2^{2k+1}}
 \binom{2k}{k},
\end{equation}
while $b_1^{(0)} = 0$. 

\subsection{Parameter $\alpha_1$ of Order $N$, $\alpha_2$ Fixed}\label{s4.4}
Here we substitute $\tilde{\alpha}_2=0$ into \eqref{4.7}, and then substitute the result in \eqref{4.8} and the case $n=2$ of \eqref{4.10} with $\alpha_1=\delta_1,\tilde{\alpha}_2=0$ to deduce the analogue of Proposition \ref{PSa}.

\begin{proposition}\label{PSb}
Define
$$
c_- =           \tilde{\alpha}_1^2/(\tilde{\alpha}_1+2)^2.
$$
With $\alpha_1=\tilde{\alpha}_1N\kappa+\delta_1$ we have
\begin{align}
W_1^0(s_1)& 
= - {\tilde{\alpha}_1 \over 2 s_1} + {\tilde{\alpha}_1 + 2 \over 2 s_1 (s_1 - 1)}
\sqrt{(s_1 - 1) ( s_1 - c_-)}
\label{4.20}
\\W_1^1(s_1)&
= \alpha_2 \bigg ( {1 \over 2 (1 - s_1)} + {1 \over 2 \sqrt{(s_1 - 1) (s_1 - c_-)} }\bigg ) +(\kappa-1)\frac{1-c_{-}}{4(s_1-1)(s_1-c_{-})}\nonumber \\
& + (\kappa - 1-\delta_1) \bigg (
{1 \over 2 s_1} +
{ \tilde{\alpha}_1 -  (\tilde{\alpha}_1 + 2) s_1 \over 2 (\tilde{\alpha}_1+2)s_1 \sqrt{(s_1 - 1)(s_1 - c_-)}} \bigg ) \label{4.21}
\end{align}
while $W_2^0(s_1,s_2)$ is given by  (\ref{2.7a}) with 
\begin{equation}\label{4.20a}
(a,b) = \left(c_-,1\right).
\end{equation}
\end{proposition}

We remark that taking the limit $\tilde{\alpha}_1 \to 0$ in the results of Proposition \ref{PSb} reclaims the
results of Proposition \ref{PSa} with $\alpha_1 = \delta_1$.

In the special cases $\beta = 1,2$ and 4 (\ref{4.20}) and (\ref{4.21}) have been calculated in \cite[Eqns.~(47b) and (108) with changes of notation]{MS12}, as corollaries of knowledge of the corresponding
moment sequence. We refer then to \cite{MS12} for the explicit forms of $b_0^{(k)}, b_1^{(k)}$ in the moment expansion
(\ref{Mn}) in this setting.

%

\subsection{Parameters $\alpha_1$ and $\alpha_2$ of Order $N$}\label{s4.5}

Proceeding as in the derivation of Propositions \ref{PSa} and \ref{PSb} gives us the explicit functional form of
$W_1^0$, $W_1^1$ and $W_2^0$ in this setting.

\begin{proposition}\label{PSc}
With $\alpha_1=\tilde{\alpha}_1N\kappa+\delta_1$ and $\alpha_2=\tilde{\alpha}_2N\kappa+\delta_2$, we have
\begin{align}
W_1^0(s_1) & = 
\frac{\tilde{\alpha}_2}{2(1-s_1)}-\frac{\tilde{\alpha}_1}{2s_1}-\frac{\tilde{\alpha}_1+\tilde{\alpha}_2+2}{2s_1(1-s_1)}\sqrt{(s_1-c_{-})(s_1-c_{+})}\label{4.24}
\\W_1^1(s_1)&=\left[\frac{\kappa-1}{2}\left(\frac{\tilde{\alpha}_1(1-s_1)}{s_1}+\frac{\tilde{\alpha}_2s_1}{1-s_1}\right)+\frac{(\tilde{\alpha}_1+\tilde{\alpha}_2)s_1-\tilde{\alpha}_1}{2}\left(\frac{\delta_1}{s_1}-\frac{\delta_2}{1-s_1}\right)+\delta_1+\delta_2+1-\kappa\right]\nonumber
\\&\qquad\times\frac{1}{(\tilde{\alpha}_1+\tilde{\alpha}_2+2)\sqrt{(s_1-c_{-})(s_1-c_{+})}}\nonumber
\\&\quad+\frac{\kappa-1}{2}\left(\frac{1}{s_1}-\frac{1}{1-s_1}-\frac{2-c_{-}-c_{+}}{2(s_1-c_{-})(s_1-c_{+})}\right)-\frac{1}{2}\left(\frac{\delta_1}{s_1}-\frac{\delta_2}{1-s_1}\right)\label{4.25}
\end{align}
while $W_2^0(s_1,s_2)$ is given by  (\ref{2.7a}) with 
\begin{equation}\label{4.26}
(a,b) =  (c_-, c_+), \qquad c_{\pm}=\frac{\tilde{\alpha}_1}{\tilde{\alpha}_1+\tilde{\alpha}_2+2}+2\frac{\tilde{\alpha}_2+1\pm\sqrt{(\tilde{\alpha}_1+1)(\tilde{\alpha}_2+1)(\tilde{\alpha}_1+\tilde{\alpha}_2+1)}}{(\tilde{\alpha}_1+\tilde{\alpha}_2+2)^2}.
\end{equation}
\end{proposition}

We remark that taking the limit $\tilde{\alpha}_2 \to 0$ in the results of Proposition \ref{PSc} reclaims the
results of Proposition \ref{PSb} with $\alpha_2 = \delta_2$.

Using (\ref{4.16}) and (\ref{4.17}) we read off from (\ref{4.24}) that
\begin{align*}
\rho_1^0(x) & = {(\tilde{\alpha}_1 + \tilde{\alpha}_2 + 2 ) \over 2 \pi x (1 - x)}
\sqrt{(x - c_1)(c_2 - x)} \chi_{c_1 < x < c_2} 
\end{align*}
which is in fact well known; see e.g.~\cite[Eq.~(149) with $u=\tilde{\alpha}_1 + 1$, $v =\tilde{\alpha}_2 + 1$]{MS12}
and references therein.



\section*{Acknowledgements}
The work of PJF was partially supported by the Australian Research Council Grant DP170102028 and
the ARC Centre of Excellence for Mathematical and Statistical Frontiers,
and that of AAR by Australian Research Council Grant DP140102613 and the 
ARC Centre of Excellence for Mathematical and Statistical Frontiers.

\appendix
\setcounter{equation}{0}
\section{}\label{appA}
Let $I(x_1,\ldots,x_n)$ denote the right hand side of the $n$\textsuperscript{th} loop equation \eqref{3.8}, and for our induction hypothesis, assume that $I(x_1,\ldots,x_m)=0$ for all $1\leq m\leq n-1$. We repeatedly use \eqref{2.6},
\begin{align*}
\overline{U}_n(x_1,J_n)=\overline{W}_n(x_1,J_n)+\sum_{\emptyset\neq J\subseteq J_n}\overline{W}_{n-|J|}(x_1,J_n\setminus J)\overline{U}_{|J|}(J),
\\J_n=(x_2,\ldots,x_n),\quad J_1=\emptyset.
\end{align*}
Subtracting $\overline{U}_{n-1}(J_n)I(x_1)$ from the right hand side of the loop equation for unconnected correlators \eqref{3.7} then leaves us with
\begin{align}
0&=\left[(\kappa-1)\frac{\partial}{\partial x_1}+\left(\frac{\alpha_1}{x_1}-1\right)\right]\left\{\overline{W}_n(x_1,J_n)+\sum_{\emptyset\neq J\subset J_n}\overline{W}_{n-|J|}\left(x_1,J_n\setminus J\right)\overline{U}_{|J|}(J)\right\}\nonumber
\\&\quad+\sum_{k=2}^n\frac{\partial}{\partial x_k}\left\{\frac{\overline{U}_{n-1}(x_1,\ldots,\hat{x}_k,\ldots,x_n)-\overline{U}_{n-1}(J_n)}{x_1-x_k}+\frac{1}{x_1}\overline{U}_{n-1}(J_n)\right\}\nonumber
\\&\quad+\kappa\overline{W}_{n+1}(x_1,x_1,J_n)+\kappa\sum_{\emptyset\neq J\subset J_n}\overline{W}_{n+1-|J|}(x_1,x_1,J_n\setminus J)\overline{U}_{|J|}(J)\nonumber
\\&\quad+\kappa\sum_{J\subseteq J_n}\overline{W}_{n-|J|}(x_1,J_n\setminus J)\overline{U}_{|J|+1}(x_1,J)-\kappa\overline{W}_1(x_1)\overline{W}_1(x_1)\overline{U}_{n-1}(J_n),\label{LE3p1}
\end{align}
where $\subset$ denotes a strict or proper subset, and $\subseteq$ denotes otherwise. Here, we have used the relation
\begin{align*}
\overline{U}_{n+1}(x_1,x_1,J_n)&=\overline{W}_{n+1}(x_1,x_1,J_n)+\sum_{\emptyset\neq J'\subseteq (x_1,J_n)}\overline{W}_{n-|J'|}(x_1,(x_1,J_n)\setminus J')\overline{U}_{|J'|}(J')
\end{align*}
and the fact that $\{\emptyset\neq J'\subseteq(x_1,J_n)\}=\{\emptyset\neq J'\subseteq J_n\}\cup\{J'=(x_1,J'')\,|\,J''\subseteq J_n\}$.

Continuing on, we note that for any $J\subseteq J_n$,
\begin{align*}
&\overline{U}_{|J|}(J)\frac{\partial}{\partial x_1}\overline{W}_{n-|J|}(x_1,J_n\setminus J)=\frac{\partial}{\partial x_1}\left\{\overline{W}_{n-|J|}(x_1,J_n\setminus J)\overline{U}_{|J|}(J)\right\},
\\
\\&\overline{U}_{|J|}(J)\frac{\partial}{\partial x_k}\frac{\overline{W}_{n-1-|J|}(x_1,(J_n\setminus J)\setminus(x_k))-\overline{W}_{n-1-|J|}(J_n\setminus J)}{x_1-x_k}
\\&=\frac{\partial}{\partial x_k}\frac{\overline{W}_{n-1-|J|}(x_1,(J_n\setminus J)\setminus(x_k))\overline{U}_{|J|}(J)-\overline{W}_{n-1-|J|}(J_n\setminus J)\overline{U}_{|J|}(J)}{x_1-x_k},\quad x_k\in(J_n\setminus J),
\\
\\&\overline{U}_{|J|}(J)\frac{\partial}{\partial x_k}\left\{\frac{1}{x_1}\overline{W}_{n-1-|J|}(J_n\setminus J)\right\}=\frac{\partial}{\partial x_k}\left\{\frac{1}{x_1}\overline{W}_{n-1-|J|}(J_n\setminus J)\overline{U}_{|J|}(J)\right\},\quad x_k\in(J_n\setminus J),
\end{align*}
as the relevant partial derivatives do not depend on any of the variables present in our $\overline{U}$ multipliers. Hence, subtracting $\sum\limits_{\emptyset\neq J\subset J_n}I(x_1,J_n\setminus J)\overline{U}_{|J|}(J)$ from our equation \eqref{LE3p1} results in
\begin{align}
0&=\left[(\kappa-1)\frac{\partial}{\partial x_1}+\left(\frac{\alpha}{x_1}-1\right)\right]\overline{W}_n(x_1,J_n)+\kappa\overline{W}_{n+1}(x_1,x_1,J_n)\nonumber
\\&\quad+\sum_{k=2}^n\frac{\partial}{\partial x_k}\left\{\frac{\overline{W}_{n-1}(x_1,\ldots,\hat{x}_k,\ldots,x_n)-\overline{W}_{n-1}(J_n)}{x_1-x_k}+\frac{1}{x_1}\overline{W}_{n-1}(J_n)\right\}\nonumber
\\&\quad+\kappa\sum_{J\subseteq J_n}\overline{W}_{n-|J|}(x_1,J_n\setminus J)\overline{U}_{|J|+1}(x_1,J)-\kappa\overline{W}_1(x_1)\overline{W}_1(x_1)\overline{U}_{n-1}(J_n)\nonumber
\\&\quad-\kappa\sum_{\emptyset\neq J\subset J_n}\sum_{K\subseteq J_n\setminus J}\overline{W}_{|K|+1}(x_1,K)\overline{W}_{n-|K|}(x_1,(J_n\setminus J)\setminus K)\overline{U}_{|J|}(J)\label{LE3p2}
\end{align}
The last two lines of the above \eqref{LE3p2} then simplifies,
\begin{align*}
&\kappa\sum_{J\subseteq J_n}\overline{W}_{n-|J|}(x_1,J_n\setminus J)\overline{U}_{|J|+1}(x_1,J)-\kappa\overline{W}_1(x_1)\overline{W}_1(x_1)\overline{U}_{n-1}(J_n)
\\&\quad-\kappa\sum_{\emptyset\neq J\subset J_n}\sum_{K\subseteq J_n\setminus J}\overline{W}_{|K|+1}(x_1,K)\overline{W}_{n-|K|}(x_1,(J_n\setminus J)\setminus K)\overline{U}_{|J|}(J)
\\&=\kappa\sum_{J\subseteq J_n}\overline{W}_{n-|J|}(x_1,J_n\setminus J)\overline{U}_{|J|+1}(x_1,J)+\kappa\sum_{K\subseteq J_n}\overline{W}_{|K|+1}(x_1,K)\overline{W}_{n-|K|}(x_1,J_n\setminus K)
\\&\quad-\kappa\sum_{J\subseteq J_n}\sum_{K\subseteq J_n\setminus J}\overline{W}_{|K|+1}(x_1,K)\overline{W}_{n-|K|}(x_1,(J_n\setminus J)\setminus K)\overline{U}_{|J|}(J),
\end{align*}
where we absorb the $\kappa\overline{W}_1(x_1)\overline{W}_1(x_1)\overline{U}_{n-1}(J_n)$ term into the latter sum, and extract the $J=\emptyset$ terms from it. Furthermore, interchanging the order of summation, the last two lines of \eqref{LE3p2} become
\begin{align*}
&\kappa\sum_{J\subseteq J_n}\overline{W}_{n-|J|}(x_1,J_n\setminus J)\overline{U}_{|J|+1}(x_1,J)+\kappa\sum_{K\subseteq J_n}\overline{W}_{|K|+1}(x_1,K)\overline{W}_{n-|K|}(x_1,J_n\setminus K)
\\&\quad-\kappa\sum_{K\subseteq J_n}\overline{W}_{|K|+1}(x_1,K)\sum_{J\subseteq J_n\setminus K}\overline{W}_{n-|J|-|K|}(x_1,(J_n\setminus K)\setminus J)\overline{U}_{|J|}(J)
\\&=\kappa\sum_{J\subseteq J_n}\overline{W}_{n-|J|}(x_1,J_n\setminus J)\overline{U}_{|J|+1}(x_1,J)+\kappa\sum_{K\subseteq J_n}\overline{W}_{|K|+1}(x_1,K)\overline{W}_{n-|K|}(x_1,J_n\setminus K)
\\&\quad-\kappa\sum_{K\subseteq J_n}\overline{W}_{|K|+1}(x_1,K)\overline{U}_{n-|K|}(x_1,J_n\setminus K),
\end{align*}
which, after writing $L=J_n\setminus K$, simplifies to
\begin{align*}
&\kappa\sum_{J\subseteq J_n}\overline{W}_{n-|J|}(x_1,J_n\setminus J)\overline{U}_{|J|+1}(x_1,J)+\kappa\sum_{K\subseteq J_n}\overline{W}_{|K|+1}(x_1,K)\overline{W}_{n-|K|}(x_1,J_n\setminus K)
\\&-\kappa\sum_{L\subseteq J_n}\overline{W}_{n-|L|}(x_1,J_n\setminus L)\overline{U}_{|L|+1}(x_1,L)
\\=&\kappa\sum_{K\subseteq J_n}\overline{W}_{|K|+1}(x_1,K)\overline{W}_{n-|K|}(x_1,J_n\setminus K).
\end{align*}
Replacing the last two lines of \eqref{LE3p2} by this result then completes the proof.

\section*{Appendix B}\label{}
\setcounter{equation}{0}
\renewcommand{\theequation}{B.\arabic{equation}}
\subsection*{Jacobi $ \beta $ Ensembles}\label{}

It was commented in the paragraph beginning with equation (\ref{4.15}) that the moments of the
Jacobi $\beta$-ensemble are rational functions in $N$, which can for low orders be computed explicitly
using Jack polynomial theory. Aspects of the latter are available as a computer algebra package \cite{DES07},
which can used to obtain the explicit form of $m_3^{(J)}$, with $m^{(J)}_{1}$ and $m^{(J)}_{2}$ already known
from \cite{MR15}. Making use of a partial fraction expansion, we find
{\small \begin{multline}
  m^{(J)}_{3} = \frac{5}{16}\,N
	+\frac{1}{32}\frac{5\,\alpha_{1}-11\,\alpha_{2}+3\,\kappa -3}{\kappa}
\\
	-\frac{1}{128}\frac{(\alpha_{1}-\alpha_{2})
  (\alpha_{1}+\alpha_{2}-2\,\kappa +2)}{{\kappa}^{3}}
    \\ \times
    \frac{\left[(\alpha_{1}-\alpha_{2})(\alpha_{1}+\alpha_{2}-2\kappa+2)-4\,\kappa \right] \left[(\alpha_{1}-\alpha_{2})(\alpha_{1}+\alpha_{2}-2\kappa+2)-8\,\kappa \right]}
         {(2\,\kappa N+\alpha_{1}+\alpha_{2}-2\,\kappa +2)}
\\
	+\frac{1}{32}\frac{(\alpha_{1}-\alpha_{2}+1)(\alpha_{1}-\alpha_{2}-1)(\alpha_{1}+\alpha_{2}-2\,\kappa+1)(\alpha_{1}+\alpha_{2}-2\,\kappa +3)}
	{\kappa (\kappa +1 )( 2\,\kappa +1)}
	\\ \times
	\frac{\left[ (\alpha_{1}-\alpha_{2})(\alpha_{1}+\alpha_{2}-2\kappa+2)-6\,\kappa -3 \right]}{( 2\,\kappa N+\alpha_{1}+\alpha_{2}-2\,\kappa +3)}
\\
	+\frac{1}{32}\frac{(\alpha_{1}-\alpha_{2}-\kappa)(\alpha_{1}-\alpha_{2}+\kappa)(\alpha_{1}+\alpha_{2}-3\,\kappa +2)(\alpha_{1}+\alpha_{2}-\kappa +2)}{{\kappa}^{3}(\kappa +1)(\kappa +2)}
    \\ \times
    \frac{\left[(\alpha_{1}-\alpha_{2})(\alpha_{1}+\alpha_{2}-2\kappa+2)-3\,{\kappa}^{2}-6\,\kappa \right]}{(2\,\kappa N+\alpha_{1}+\alpha_{2}-3\,\kappa +2)}
\\
	-\frac{1}{128}\frac{(\alpha_{1}-\alpha_{2}+2\,\kappa)(\alpha_{1}-\alpha_{2})(\alpha_{1}-\alpha_{2}-2\,\kappa)}
                       {{\kappa}^{3}(\kappa +1)(2\,\kappa +1)}
    \\ \times
    \frac{(\alpha_{1}+\alpha_{2}+2)(\alpha_{1}+\alpha_{2}-2\,\kappa +2)(\alpha_{1}+\alpha_{2}-4\,\kappa +2)}{(2\,\kappa N+\alpha_{1}+\alpha_{2}-4\,\kappa +2)}
\\
	-\frac{1}{128}\frac{(\alpha_{1}-\alpha_{2}-2)(\alpha_{1}-\alpha_{2})(\alpha_{1}-\alpha_{2}+2)}
                       {\kappa(\kappa +1)(\kappa +2)}
     \\ \times
     \frac{(\alpha_{2}+\alpha_{1}-2\,\kappa)(\alpha_{1}+\alpha_{2}-2\,\kappa +4)(\alpha_{1}+\alpha_{2}-2\,\kappa +2)}{(2\,\kappa N+\alpha_{1}+\alpha_{2}-2\,\kappa +4)} .
\end{multline}}

\subsection*{Dyson Circular $ \beta $ Ensemble}\label{}

It follows from \cite[Prop.~3.9.1]{Fo10} that
the specialisation $ \alpha_{1} \mapsto -\kappa N+\kappa -1 $, $ \alpha_{2} \mapsto 0 $ yields the moments 
$m^{(C)}_{l}$
of the 
connected two-point correlation (see \cite[Eq.~(1.12)]{WF14}) for the Dyson Circular $ \beta $ Ensemble
according to the prescription
\begin{equation}
     M_{l} = \left.\frac{m^{(J)}_{l}(\kappa, N, \alpha_{1}, \alpha_{2})}{(\kappa^{-1}\alpha_{1}+\kappa^{-1}+N-1)}\right|_{\alpha_{1} \mapsto -\kappa N+\kappa -1,\alpha_{2} \mapsto 0} , \qquad
      M_l = \frac{\kappa}{l}\left( m_l^{(C)}(N,\kappa)+N \right).
\end{equation}
This allows us to reclaim  the equations for $M_1$ and $M_2$ of \cite[Eq.~(4.17) \& (4.18)]{WF14}
and further allows us to compute, using knowledge of $m^{(J)}_{3}$, that
\begin{multline}
  M_{3} = 1+    
		  \frac{(\kappa-1)}{(\kappa N-\kappa +1)}
\\
		 -4\,\frac{\kappa(\kappa -1)(\kappa-2)}{(\kappa +1)(2\,\kappa +1)(\kappa N-\kappa +2)}
		 +4\,\frac{(\kappa-1)(2\,\kappa -1)}{(\kappa +1)(\kappa +2)(\kappa N-2\,\kappa +1)}
\\
		 +\frac{(\kappa-1)(\kappa -2)(\kappa -3)}{(\kappa +1)(\kappa +2) (\kappa N-\kappa +3)}
		 +\frac{( \kappa -1)(2\,\kappa-1)(3\,\kappa -1)}{(\kappa +1)(2\,\kappa +1)(\kappa N-3\,\kappa +1)}.
\label{}
\end{multline}


\begin{thebibliography}{10}

\bibitem{AJM90}
J.~Ambj{\o}rn, J.~Jurkiewicz, and Yu. Makeenko, \emph{Multiloop correlators for
  two-dimensional quantum gravity}, Phys. Lett. B \textbf{251} (1990), 517--524.

\bibitem{AM90}
J.~Ambj{\o}rn and Yu.M. Makeenko, \emph{Properties of loop equations for the
  {H}ermitian matrix model and for two-dimensional quantum gravity}, Mod.
  Phys. Lett. \textbf{5} (1990), 1753--1763.

\bibitem{Ao87}
K.~Aomoto, \emph{Jacobi polynomials associated with {Selberg's} integral}, SIAM
  J. Math. Analysis \textbf{18} (1987), 545--549.

\bibitem{Be97}
C.W.J. Beenakker, \emph{Random-matrix theory of quantum transport}, Rev. Mod.
  Phys. \textbf{69} (1997), 731--808.
  
 \bibitem{BFB97}
 P. W. Brouwer, K. M. Frahm, and C. W. J. Beenakker, \emph{Quantum mechanical time-delay
matrix in chaotic scattering}, Phys. Rev. Lett. \textbf{78} (1997), 4737--4740.

\bibitem{BEMN10}
G.~Borot, B.~Eynard, S.N. Majumdar, and C.~Nadal, \emph{Large deviations of the
  maximal eigenvalue of random matrices}, J. Stat. Mech. \textbf{2011} (2011),
  P11024.

\bibitem{BG12}
G.~Borot and A.~Guionnet, \emph{Asymptotic expansion of $\beta$ matrix models
  in the one-cut regime}, Commun. Math. Phys. \textbf{2013} (2013), 447--483.

\bibitem{BMS11}
A.~Brini, M.~Mari\ {n}o, and S.~Stevan, \emph{The uses of the refined matrix
  model recursion}, J. Math. Phys. \textbf{52} (2011), 35--51.

\bibitem{CMSV16a}
F.D. Cunden, F.~Mezzadri, N.~Simm, and P.~Vivo, \emph{Correlators for the
  {W}igner-smith time-delay matrix of chaotic cavities}, J. Phys. A \textbf{49}
  (2016), 18LT01.

\bibitem{CMSV16b}
\bysame, \emph{Large-{$N$} expansion for the time-delay matrix of chaotic
  cavities}, J. Math. Phys. \textbf{57} (2016), 111901.

\bibitem{DF06}
P.~Desrosiers and P.J. Forrester, \emph{Hermite and {L}aguerre
  $\beta$-ensembles: asymptotic corrections to the eigenvalue density}, Nucl.
  Phys. B \textbf{743} (2006), 307--332.

\bibitem{DE02}
I.~Dumitriu and A.~Edelman, \emph{Matrix models for beta ensembles}, J. Math.
  Phys. \textbf{43} (2002), 5830--5847.

\bibitem{DE05}
\bysame, \emph{Global spectrum fluctuations for the $\beta$-{H}ermite and
  $\beta$-{L}aguerre ensembles via matrix models}, J. Math. Phys. \textbf{47}
  (2006), 063302.
  
  \bibitem{DES07}
I.~Dumitriu, A.~Edelman, and G.~Shuman, \emph{Mops: Multivariate orthogonal
  polynomials (symbolically)}, J. Symb. Comput. \textbf{42} (2007), 587--620.

\bibitem{DP12}
I.~Dumitriu and E.~Paquette, \emph{Global fluctuations for linear statistics of
  $\beta$ {J}acobi ensembles}, Random Matrices: Theory Appl. \textbf{01}
  (2012), 1250013.

\bibitem{Fo10}
P.J. Forrester, \emph{Log-gases and random matrices}, Princeton University
  Press, Princeton, NJ, 2010.

\bibitem{FFG06e}
P.J. Forrester, N.E. Frankel, and T.M. Garoni, \emph{Asymptotic form of the
  density profile for {G}aussian and {L}aguerre random matrix ensembles with
  orthogonal and symplectic symmetry}, J. Math. Phys. \textbf{47} (2006),
  023301.

\bibitem{FW07p}
P.J. Forrester and S.O. Warnaar, \emph{The importance of the {S}elberg
  integral}, Bull. Am. Math. Soc. \textbf{45} (2008), 489--534.
  
  \bibitem{FLD16} 
  Y.V.~Fyodorov and P.~Le Doussal, \emph{Moments of the position of the maximum for GUE characteristic polynomials and for log-correlated Gaussian processes},  J Stat Phys \textbf{164} (2016), 190--240.

\bibitem{GFF05}
T.M. Garoni, P.J. Forrester, and N.E. Frankel, \emph{Asymptotic corrections to
  the eigenvalue density of the {GUE} and {LUE}}, J. Math. Phys. \textbf{46}
  (2005), 103301.

\bibitem{GMW98}
T.~Guhr, A.~M\"uller-Groeling, and H.A. Weidenm\"uller, \emph{Random matrix
  theories in quantum physics: common concepts}, Phys. Rep. \textbf{299}
  (1998), 189--425.

\bibitem{KN04}
R.~Killip and I.~Nenciu, \emph{Matrix models for circular ensembles}, Int.
  Math. Res. Not. \textbf{50} (2004), 2665--2701.

\bibitem{LV11}
G.~Livan and P.~Vivo, \emph{Moments of {W}ishart-{L}aguerre and {J}acobi
  ensembles of random matrices: application to the quantum transport problem in
  chaotic cavities}, Acta Phys. Pol. B \textbf{42} (2011), 1081--1104.
  
\bibitem{Ma14}  
O.~Marchal, \emph{Elements of proof for conjectures of {W}itte and {F}orrester about the
combinatorial structure of {G}aussian $\beta$ ensembles}, J. High Energy Phys. (2014)
2014: 3. doi:10.1007/JHEP09(2014)003

\bibitem{MR15}
F.~Mezzadri, A.K. Reynolds and B.~Winn, \emph{Moments of the eigenvalue densities and of
  the secular coefficients of $\beta$-ensembles}, Nonlinearity {\bf 30} (2015), 1034.

\bibitem{MS11}
F.~Mezzadri and N.J. Simm, \emph{Moments of the transmission eigenvalues,
  proper delay times and random matrix theory {I}}, J. Math. Phys.
  \textbf{52} (2011), 103511.

\bibitem{MS12}
\bysame, \emph{Moments of the transmission eigenvalues, proper delay times and
  random matrix theory {II}}, J. Math. Phys. \textbf{53} (2012), 053504.

\bibitem{Mi83}
A.A. Migdal, \emph{Loop equations and {$1/N$} expansions}, Phys. Rep.
  \textbf{102} (2004), 199--290.

\bibitem{Mu82}
R.J. Muirhead, \emph{Aspects of multivariate statistical theory}, Wiley, New
  York, 1982.

\bibitem{No08}
M.~Novaes, \emph{Statistics of quantum transport in chaotic cavities with
  broken time reversal symmetry}, Phys. Rev. B \textbf{78} (2008), 035337.

\bibitem{PS11}
L.~Pastur and M.~Shcherbina, \emph{Eigenvalue distribution of large random
  matrices}, American Mathematical Society, Providence, RI,, 2011.

\bibitem{Ra16}
A.A. Rahman, \emph{Moments of the {L}aguerre $\beta$ ensembles}, MSc.~thesis,
  The University of Melbourne, 2016.

\bibitem{Se44}
A.~Selberg, \emph{Bemerkninger om et multipelt integral}, Norsk. Mat. Tidsskr.
  \textbf{24} (1944), 71--78.

\bibitem{St99}
R.P. Stanley, \emph{Enumerative combinatorics}, vol.~2, Cambridge University
  Press, Cambridge, 1999.

\bibitem{VV08}
P.~Vivo and E.~Vivo, \emph{Transmission eigenvalue densities and moments in
  chaotic cavities from random matrix theory}, J. Phys. A \textbf{41} (2008),
  122004.

\bibitem{WF14}
N.S. Witte and P.J. Forrester, \emph{Moments of the {G}aussian $\beta$ ensembles
  and the large {$N$} expansion of the densities}, J. Math. Phys. \textbf{55}
  (2014), 083302.

\bibitem{WF15}
\bysame, \emph{Loop equation analysis of the circular ensembles}, J. High Energy Phys.
  \textbf{2015} (2015), 173.

\end{thebibliography}

\providecommand{\bysame}{\leavevmode\hbox to3em{\hrulefill}\thinspace}
\providecommand{\MR}{\relax\ifhmode\unskip\space\fi MR }
\providecommand{\MRhref}[2]{%
  \href{http://www.ams.org/mathscinet-getitem?mr=#1}{#2}
}
\providecommand{\href}[2]{#2}

\end{document}